\documentclass{article} 
\pdfoutput=1

\usepackage{amsmath, amsthm, enumerate, amssymb, mathrsfs, authblk}
\usepackage{microtype} 
\usepackage{multicol} 
\usepackage{fullpage} 
\usepackage{tikz}
\usepackage{tikz-cd}
\usetikzlibrary{positioning, calc}

\usepackage[all]{xy}

\usepackage[osf,sc]{mathpazo}

\usepackage[backend=bibtex, maxnames=5, url=false, doi=false]{biblatex}
\bibliography{SUSYQM}

\usepackage{mathtools,amsmath,amsthm,commath,braket}
\usepackage{subcaption}
\newtheorem{theorem}{Theorem}
\theoremstyle{definition}

\newtheorem{example}{Example}
\DeclareMathOperator{\im}{im}
\DeclareMathOperator{\pr}{pr}
\DeclareMathOperator{\tr}{tr}
\DeclareMathOperator{\SU}{SU}
\newcommand\SUl{\SU(2)_{\text{left}}}
\newcommand\SUr{\SU(2)_{\text{right}}}
\newcommand\Left{_{\text{left}}}
\newcommand\Right{_{\text{right}}}
\DeclareMathOperator{\SO}{SO}
\DeclareMathOperator{\SL}{SL}
\DeclareMathOperator{\Spin}{Spin}
\newcommand\so{\mathfrak{so}}
\newcommand\su{\mathfrak{su}}

\newcommand\conj{^{\mathsf c}}
\newcommand\CP[1]{\mathbb{C}\mathrm P^{#1}}
\newcommand\rep[1]{\mathbf{#1}}

\newcommand\drawdiamond[1]{
\begin{scope}
\clip (0,{#1}) -- ({#1},0) -- (0,{-1*#1}) -- ({-1*#1},0) -- cycle;
\draw[gray, dotted] ({-1*#1},{-1*#1}) grid[dotted] ({#1},{#1}); 
\end{scope}
\draw[gray, dotted] (0,{#1}) -- ({#1},0) -- (0,{-1*#1}) -- ({-1*#1},0) -- cycle;
}

\usepackage[capitalize]{cleveref} 
\begin{document}

\title{Real homotopy theory and\\ supersymmetric quantum mechanics}
\author[1,2]{Hyungrok Kim}
\author[1]{Ingmar Saberi}
\affil[1]{{\small Walter Burke Institute for Theoretical Physics, Caltech, Pasadena, California 91125, USA}}
\affil[2]{{\small School of Natural Sciences, Institute for Advanced Study, Princeton, New Jersey 08540, USA}}
\date{\today \\[1ex]  {\small Report No. CALT-TH 2015-054}}
\maketitle

\begin{abstract}
In the context of studying string backgrounds, much work has been devoted to the question of how similar a general quantum field theory (specifically, a two-dimensional superconformal theory) is to a sigma model. Put differently, one would like to know how well or poorly one can understand the physics of string backgrounds in terms of concepts of classical geometry. Much attention has also been given of late to the question of how geometry can be encoded in a microscopic physical description that makes no explicit reference to space and time. We revisit the first question, and review both well-known and less well-known results about geometry and sigma models from the perspective of dimensional reduction to supersymmetric quantum mechanics. The consequences of arising as the dimensional reduction of a \(d\)-dimensional theory for the resulting quantum mechanics are explored. 
In this context, we reinterpret the minimal models of rational (more precisely, complex) homotopy theory as certain supersymmetric Fock spaces, with unusual actions of the supercharges. The data of the Massey products appear naturally as supersymmetric vacuum states that are entangled between different degrees of freedom. 
This connection between entanglement and geometry is, as far as we know, not well-known to physicists. 
In addition, we take note of an intriguing numerical coincidence in the context of string compactification on hyper-K\"ahler eight-manifolds.
\end{abstract}

\begin{multicols}{2}

\section*{Introduction}

Supersymmetric quantum mechanics, which describes a supersymmetric particle moving on a compact Riemannian manifold, has been studied by many authors (see, for example, \cite{Witten-Morse,AlvarezGaume,Witten-Constraints}) and has deep connections to the geometry and topology of the target space~\(M\) on which the particle moves. The theory is a supersymmetric sigma model in \(0+1\)~dimensions. Its Hilbert space is the space of complex-valued differential forms on~\(M\), equipped with the inner product that generalizes the \(L^2\) inner product from undergraduate quantum mechanics:
\begin{equation}
(\alpha, \beta) = \int_M \star \bar\alpha \wedge \beta. 
\label{innerproduct}
\end{equation}
The bar here denotes complex conjugation. The \(\mathcal N=1\) supersymmetry algebra in $0+1$ dimensions~\cite{AlvarezGaume} consists of the exterior derivative \(\dif{}\) and its adjoint \(\dif{}^\dagger\) with respect to the inner product~\protect\eqref{innerproduct}, as well as their anticommutator, the Laplacian (which is the Hamiltonian operator of the theory, and commutes with the supercharges). With respect to the grading, \(\dif{}\) carries degree one. 

We will sometimes refer to this collection of operators as the de~Rham algebra. Representations of this algebra are of two types: there is the standard or ``long'' representation, consisting of two generators of adjacent degree that are mapped to one another by the supercharges, and there are ``short'' one-dimensional representations that are annihilated by both~\(\dif{}\) and~$\dif{}^\dagger$. The ``BPS bound'' in this theory is simply $\Delta\ge 0$. It follows from the equation \(\Delta = \{ \dif{}, \dif{}^\dagger\}\), and it is easily shown that a representation is short if and only if it has zero energy.
We draw a picture of the spectrum of the theory in \protect\cref{cohomology}.

\begin{figure*}
\begin{center}
\begin{tikzpicture}[x=2.5cm,y=2cm]
\draw [->] (0.5, 0.8) -- (0.5,2) node[left] {\(\Delta\)} -- (0.5,3);
\draw (0.5,1) ++(-0.1,0) -- ({0.5+0.1},1);

\coordinate (E1) at (2,2);
\coordinate (E2) at (2.5,2);

\fill [black] (E1) circle (2pt);
\fill [black] (E2) circle (2pt);
\draw [->] (E1)++(0.1,0.1) to [out=30, in=150] node[above]{\(\scriptstyle\dif{}\)} ($(E2)+(-0.1,0.1)$);
\draw [<-] (E1)++(0.1,-0.1) to [out=-30, in=210] node[below]{\(\scriptstyle\dif{}^\dagger\)} ($(E2)+(-0.1,-0.1)$);

\coordinate (E3) at (1,2.5);
\coordinate (E4) at (1.5,2.5);

\fill [black] (E3) circle (2pt);
\fill [black] (E4) circle (2pt);
\draw [->] (E3)++(0.1,0.1) to [out=30, in=150] node[above]{\(\scriptstyle\dif{}\)} ($(E4)+(-0.1,0.1)$);
\draw [<-] (E3)++(0.1,-0.1) to [out=-30, in=210] node[below]{\(\scriptstyle\dif{}^\dagger\)} ($(E4)+(-0.1,-0.1)$);

\coordinate (V1) at (1,1);
\coordinate (V2) at (2,1);
\coordinate (V3) at (3,1);
\fill [black] (V1) circle (2pt);
\fill [black] (V2) circle (2pt);
\fill [black] (V3) circle (2pt);

\coordinate (Vtext) at (3.5, 1);
\node[right] at (Vtext) {\(\left.\vphantom{\rule{0pt}{0.7cm}}\right\}\) vacua};

\coordinate (Etext) at (3.5, 2.25);
\node[right] at (Etext) {\(\left.\vphantom{\rule{0pt}{1cm}}\right\}\) excited states};
\end{tikzpicture}
\end{center}
\caption[The de~Rham complex of a compact manifold]{A schematic diagram of the de~Rham complex of a compact manifold, drawn to emphasize the connections with supersymmetry. The vertical axis is the energy (eigenvalue of the Laplacian), and the horizontal axis is the homological degree. The compactness assumption, which ensures that the spectrum is discrete, is crucial and cannot be relaxed.}\label{cohomology}
\end{figure*}
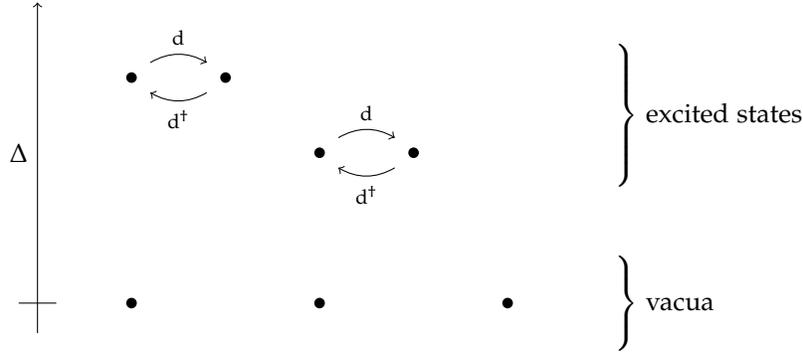

The BPS states of the theory also admit a description as the cohomology of a chosen supercharge. There is only one supercharge to choose, namely $\dif{}$; it is a scalar, since we are in $0+1$~dimensions, and so we can think of passing to its cohomology as a topological twist of the theory.

This story continues in many directions: for instance, one can imagine adding a superpotential term to the action. The resulting physics has a beautiful description in terms of the WKB approximation and the Morse complex of~\(M\). However, we will not continue down this path, and refer the interested reader to~\cite{Witten-Morse}. Instead, we wish to emphasize some points of differential topology which are less well known to most physicists.

When one studies a space~\(M\) using an algebraic invariant, it is natural to ask how well that invariant distinguishes inequivalent spaces. Many different manifolds share the same Euler characteristic, while the simplicial chain complex of~\(M\) has perfect information: from it we can reconstruct \(M\) up to homeomorphism by gluing together a simplicial complex. Of course, there are many different simplicial complexes corresponding to the same manifold; in this sense, despite containing perfect information, a triangulation of~$M$ is not an invariant at all.

The cohomology of~$M$ is an honest invariant, lying somewhere between these two in strength: it determines the Euler characteristic and is determined by a triangulation. But how close is it to a perfect invariant of~\(M\)?

In physics, the simplicial chain complex of the target space \(M\) does not arise naturally. However, the Hilbert space associated to~\(M\) (i.e.\ its de~Rham complex) plays a similar role. We are therefore led to ask: How much information about a manifold can be recovered from its de~Rham complex? And how much information about the de~Rham complex is contained in its cohomology---or, in physical language, how much can we learn about a supersymmetric quantum mechanics by studying its BPS spectrum?

To address these questions, we recall some structures and results of classical algebraic topology. (For more background, see~\cite{Bott-Tu}.)
Cohomology is an invariant of homotopy type: a space that is homotopy-equivalent to~\(M\) will have identical cohomology. For example, all the spaces \(\mathbb R^n\) for every~\(n\) have the homotopy type of a point. However, assuming that~\(M\) is compact without boundary, one can also deduce its dimension from its cohomology through Poincar\'e duality. One therefore has a way of distinguishing between homotopy-equivalent spaces such as \(M\) and \(M\times\mathbb R\). 

A homotopy equivalence between~\(M\) and another space~\(N\) induces a relationship (also called homotopy equivalence) between their respective de~Rham complexes, considered as commutative differential graded algebras. Homotopy-equivalent CDGAs have identical cohomology, but the converse statement is not true. To see why, we need to consider algebraic operations defined on cohomology classes. 

The de~Rham complex is an algebra. Its multiplication is the wedge product of differential forms. This product descends to cohomology classes, which therefore form an algebra. It also (in more subtle fashion) induces \emph{additional} structures: the cohomology carries a set of higher operations known as \emph{Massey products}. These are partially defined functions and, where they are defined, may or may not vanish. Roughly speaking, the Massey products record ambiguities about how products of certain harmonic forms become exact. 
Another intuitive picture is that Massey products are like higher-order linking numbers: they detect or measure entangledness between sets of three or more cycles which are pairwise unlinked, like the Borromean rings (\protect\cref{Borromeo}). Here we have been sloppy and identified cocycles with their Poincar\'e-dual homology cycles.

In the simplest case, the Massey triple product is defined as follows.
Let \(u\), \(v\), and \(w\) be representatives of three nontrivial cohomology classes in~$\mathrm H^\bullet(M)$, of homogeneous degree, and let $\bar u = (-)^{F+1} u$, where \(F\) denotes the degree operator. (This convention helps keep track of signs.) Furthermore, suppose that the pairwise products \([u][v] = [v][w] = 0\); in other words, that \(\bar u\wedge v = \dif s\) for some form~\(s\), and similarly \(\bar v \wedge w = \dif t\).The Massey product of \(u\), \(v\), and \(w\) is then given by the following expression:
\begin{equation}
m(u,v,w) = [ \bar s \wedge w + \bar u\wedge t].
\end{equation}
It is simple to check that the representative \(\bar s\wedge w + \bar u \wedge t\) is closed. For the Massey product of cohomology classes to be well defined, one should check that the choice of representatives \(u\), \(v\), and \(w\) does not affect the end result. In fact, the end result is not exactly invariant, but it becomes well-defined as an element of a particular quotient of the cohomology.

If one can make a consistent choice of representatives for cohomology classes of~\(M\) such that all higher Massey products simultaneously vanish, one says that~\(M\) is \emph{formal}. More transparently, a space is formal if its de~Rham complex is homotopy-equivalent to its cohomology ring, viewed as a CDGA with zero differential. Passing to cohomology thus loses information about the homotopy type of a manifold exactly when that manifold fails to be formal. This  is an answer to our second question above. 

As to the first question, the seminal work of Sullivan~\cite{Sullivan-Infinitesimal} and Barge~\cite{Barge} shows that, up to ``finite ambiguities,'' the diffeomorphism type of a simply connected smooth compact manifold is determined by the rational homotopy type of its de~Rham complex.\footnote{To be precise, we here mean the $\mathbb Q$-polynomial variant of the de~Rham complex, as defined by Sullivan.} Moreover, given a choice of Riemannian metric on the manifold, this homotopy type can be presented canonically by a finitely generated ``minimal model'' CDGA over~$\mathbb Q$. We give a more detailed exposition of Sullivan's methods and results later in the paper. For now, the reader should bear in mind that the de~Rham complex is, in some appropriate sense, ``almost'' a perfect algebraic model like a triangulation. 

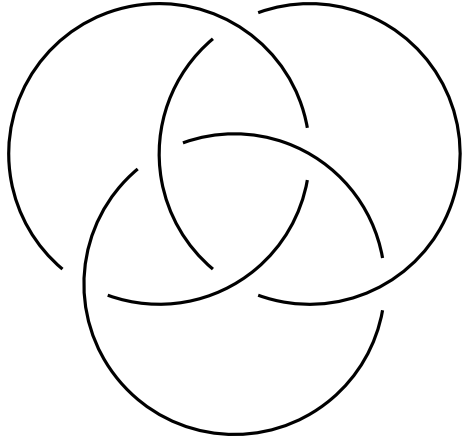
\begin{figure*}
\begin{center}
\begin{tikzpicture}[x=2cm,y=2cm]

\draw [black, very thick] (2,2) ++(10:1) arc [radius=1, start angle=10, end angle=230];
\draw [black, very thick] (2,2) ++(250:1) arc [radius=1, start angle=250, end angle=350];

\draw [black, very thick] (3,2) ++(-110:1) arc [radius=1, start angle=-110, end angle=110];
\draw [black, very thick] (3,2) ++(130:1) arc [radius=1, start angle=130, end angle=230];

\draw [black, very thick] (2.5,{2-0.5*sqrt(3)}) ++(10:1) arc [radius=1, start angle=10, end angle=110];
\draw [black, very thick] (2.5,{2-0.5*sqrt(3)}) ++(130:1) arc [radius=1, start angle=130, end angle=350];
\end{tikzpicture}
\end{center}
\caption[The Borromean rings]{The Borromean rings.}\label{Borromeo}
\end{figure*}

One goal of this paper is to understand how these results relate to physics. To make contact with Sullivan's work, we need a physical context in which the algebraic structure of the de~Rham complex, or at least of the cohomology, arises. Supersymmetric quantum mechanics, by itself, does not provide such a context: the wedge product differential forms is just a variant of multiplying two wavefunctions together, which is not \emph{a~priori} meaningful in quantum mechanics. Algebraic structures do occur in physics, but they are more naturally connected to operators. As such, one needs a physical setting in which there is a state-operator correspondence. 

Two-dimensional conformal field theory is the simplest such setting, and two-dimensional superconformal sigma models are well studied. In order for a 2d sigma model to admit a topological twist, it must have enhanced supersymmetry; this occurs precisely when the target space is K\"ahler.

K\"ahler manifolds also play a special role in work of Sullivan and collaborators on rational homotopy theory; the paper \cite{DGMS} proves that every compact K\"ahler manifold is formal. Moreover, this fact follows from a simple identity relating the differential operators \(\partial\) and~\(\bar\partial\) on the de~Rham complex, called the \(\dif\dif{}\conj\)-lemma. The same identity is responsible for supersymmetry enhancement in K\"ahler sigma models; a version of the \(\dif\dif{}\conj\)-lemma can be proved for quantum mechanical systems obtained from dimensional reduction from field theories with enhanced supersymmetry. We will expand on this point later. 

There are strong constraints on the topology of K\"ahler manifolds. Formality is one of these, but much more can be said. To what extent enhanced supersymmetry imposes analogous constraints on two-dimensional \(\mathcal N=(2,2)\) SCFTs has been studied by many authors, notably~\cite{LVW,FOFKS}. A second aim of this paper is to revisit and clarify the analogies between K\"ahler geometry and \(\mathcal N=(2,2)\) SCFT. For us, the key ingredient is always supersymmetric quantum mechanics: we study the $0+1$-dimensional theories obtained by dimensional reduction from supersymmetric field theories in $d\ge 2$. In theories that are sigma models, this supersymmetric quantum mechanics is precisely the de~Rham complex of the target space, or (as in the case of the $B$-twist of the two-dimensional sigma model) some variant thereof.
Furthermore, extra structures arise when a \(d\)-dimensional theory arises as the dimensional reduction of a $d'$-dimensional theory, $d'>d$; this was considered by~\cite{FOFKS}. For us, $d = 1, 2, 4, 6$; as the dimension grows, the number of supercharges (and also the number of bosonic symmetries) increases accordingly. A table outlining these hierarchies can be found later in the paper~(\protect\cref{hierarchy}).

To reiterate: we review the well-known hierarchies of increasingly rich structures that appear in the following four contexts:
\begin{itemize}
\item target-space geometry that is generic, K\"ahler, and hyper-K\"ahler;\footnote{Here ``Calabi-Yau'' should appear strictly in between ``K\"ahler'' and ``hyper-K\"ahler.''}
\item supersymmetric quantum mechanics, with one, two, and four supercharges;
\item two-dimensional superconformal field theories, with \(\mathcal N=(1,1)\), \( (2,2) \), and $(4,4)$;
\item minimally supersymmetric field theories in dimensions two, four, and six.
\end{itemize}
We explore similarities and differences between the structures that emerge in each case. As the structure becomes richer, the differences between the various categories become fewer and fewer.

We hope that this paper will also serve a pedagogical purpose, being of some use to those who wish to learn the well studied subjects we review. In studying the literature on supersymmetric sigma models and K\"ahler geometry, certain unifying themes became apparent to us, which we felt were not adequately spelled out in existing references. We hope that at least some of this thematic unity comes across in our treatment. 

Throughout, we are motivated by the following series of increasingly speculative ideas and questions: The de~Rham complex is an algebraic model of a space, which naturally arises in physics as the Hilbert space of supersymmetric quantum mechanics. Physics also offers examples of other kinds of algebraic models for spaces: whenever a sigma model (in any dimension and with any amount of supersymmetry) can be defined on a target space \(M\), its spaces of states or operators constitute an algebraic model of~\(M\). This algebraic model should always reduce to some variant of the de~Rham complex upon dimensional reduction. 

In the case of the two-dimensional \(\mathcal N=(2,2)\) sigma models defined when~\(M\) is a Calabi-Yau threefold, this algebraic model does not retain perfect information about~\(M\) or even about its cohomology. However, this is far from a failure of these ideas; the ambiguity in recovering \(M\) from this algebraic model leads precisely to the phenomenon of mirror symmetry~\cite{Shatashvili-Vafa}!

Given an algebraic model of some kind, it is natural to ask for sufficient conditions that ensure that it is the model of some space. One can get some feel for this by enumerating all possible identifiable structures that are present on algebraic models of honest spaces; these are then necessary conditions to be the model of a space. In physics, one often speaks about this problem in the language of ``geometric'' and ``non-geometric'' theories. 

The Sullivan--Barge theorem says that, with appropriate conditions, the list of conditions and structures that are \emph{necessary} for a CDGA to algebraically model an honest space are also \emph{sufficient}. As such, a motivating problem for the line of work we have pursued would be to understand the algebraic structures corresponding to (some class of) physical quantum field theories, to develop the homotopy theory of such structures,\footnote{A ``homotopy between two physical theories'' should be understood as a one-parameter family of theories interpolating between the two; that is to say, a path from one to the other in the appropriate moduli space.} and then formulate sufficient conditions for such a theory to be describable as a sigma model. If one could do this rigorously and give a statement, analogous to Sullivan--Barge, allowing one to reconstruct the target space from its algebraic model, one would almost surely have a clear understanding of the phenomenon of mirror symmetry.
While we are of course far from doing any of this in this modest paper (and others~\cite{Zhou} have previously and more expertly drawn connections between mirror symmetry and rational homotopy theory), it is our hope that the ideas we sketch here will prove useful for others to think about and eventually bear fruit of this kind.

\section{Supersymmetry algebras and their representations 
in one spacetime dimension
}

The extended supersymmetry algebra in \(0+1\) dimensions~\cite{AlvarezGaume} is as follows:
\begin{equation}
\{ Q_i, Q_j^\dagger \} = 2 \delta_{ij} H, \quad \{Q_i, Q_j\} = 0. 
\end{equation}
From these basic commutation relations, it follows that all supercharges commute with \(H\) (the Hamiltonian of the theory). The indices range from~\(1\) to~\(\mathcal N\) for \(\mathcal N\)-extended supersymmetry; although there is no restriction on~\(\mathcal N\) in principle, in cases relevant to either physics or geometry, \(\mathcal N\) is usually a small power of two.

This algebra is usually supplanted with either an operator \(F\), defining fermion number or homological degree and taken to have integer eigenvalues, or $(-)^F$, which defines fermion number modulo two. In either case the supercharges $Q_i$ should carry one unit of fermion number. In the former case, a representation of the algebra with~\(\mathcal N=1\) becomes a chain complex of Hilbert spaces; in the latter, it is a $\mathbb Z/2\mathbb Z$-graded chain complex. The reader will no doubt have noticed that the \(\mathcal N=1\) algebra is just what we called the de~Rham algebra in the introduction.

We have already mentioned the familiar classification of irreducible representations of this algebra. To remind the reader, they can be labeled with a single nonnegative number, the eigenvalue of~$H$ (as well as an integer labeling the degree). When this number (the energy) is positive, the representation is ``long'' and consists of two generators of adjacent degree, which are mapped to one another by the supercharges; when it is zero, the representation is one-dimensional (``short''). The key picture to keep in mind is \protect\cref{cohomology}. The reader will have no trouble drawing the analogous picture of a $\mathbb Z/2\mathbb Z$-graded complex.

Thanks to formal Hodge theory, this classification persists for representations of the extended supersymmetry algebra. The bosonic operators consist of $H$ together with whatever degree operators $F_1,\ldots,F_\mathcal N$ are defined. If we insist that the~$F_i$ commute, representations will be joint eigenspaces for all of these, and can be labeled by their energy together with their multidegree. Short representations will still be one-dimensional; long representations will now have dimension $2^\mathcal N$, and they will consist of generators at each corner of an \(\mathcal N\)-dimensional cube in the degree space. The cohomology with respect to any supercharge $Q_i$ is the same; it counts the short representations, and cohomology classes therefore carry a well-defined multidegree. A long representation of \(\mathcal N=2\) is shown in the $(F_1,F_2)$-plane below:
\begin{equation}
\begin{tikzcd}
\bullet \arrow{r}{Q_1} & \bullet \\
\bullet \arrow{r}{Q_1} \arrow{u}{Q_2} & \bullet \arrow{u}[swap]{Q_2}
\end{tikzcd}
\label{square}
\end{equation}

Let us focus on the case \(\mathcal N=2\), which will be important in what follows. A crucial observation is that any representation of this algebra actually admits a \(\CP1\) family of \emph{different} actions (if we relax some requirements about the existence of degree operators). Define the parameterized supercharge
\begin{equation}
Q_t = t_1 Q_1 + t_2 Q_2,
\end{equation}
where the parameter $t\in \CP1$, and we have chosen a representative such that $|t_1|^2 + |t_2|^2 = 1$. Then it is an easy calculation to show that 
\begin{equation*}
\{ Q_t, Q_t^\dagger\} = 2\left( |t_1|^2 + |t_2|^2 \right) H = 2 H.
\end{equation*}
Further, let $\smash{\tilde{t}} = [-\bar{t}_2:\bar{t}_1]$. Then $Q_{\smash{\tilde{t}}}$ is another supercharge, and 
\begin{equation*}
\{ Q_t, Q_{\smash{\tilde{t}}}\} = \{ Q_t, {Q_{\smash{\tilde{t}}}}^\dagger \} = 0. 
\end{equation*}
These two supercharges therefore define another \(\mathcal N=2\) algebra. Further, they each carry total fermion number $F_1 + F_2 = 1$. However, they are no longer eigenstates of $F_1-F_2$; conjugation with $F_1-F_2$ acts on the \(\CP1\) parameter space by the rule 
\begin{equation*}
[t_1:t_2] \mapsto [t_1:-t_2],
\end{equation*}
and the two fixed points correspond to our two original supercharges. $Q_t$ cohomology is the same for all values of~$t$; since it describes the zero-energy spectrum of the same Hamiltonian, there is no way it can change.
(For physicist readers, we are just describing the consequences of $\SU(2)$ $R$-symmetry.)

There are two variants of this algebra that emerge naturally in geometry and physics, which we will now describe. These are in some sense intermediate between \(\mathcal N=1\) and \(\mathcal N=2\) supersymmetry. The first is what we will call the \(\mathcal N=1^2\) algebra: it consists of two mutually commuting copies of the de~Rham algebra, and so is equivalent to an \(\mathcal N=2\) algebra in which the condition that the two supercharges square to the \emph{same} Hamiltonian has been relaxed. In equations,
\begin{equation}
\{ Q_i, Q_j^\dagger \} = 2 \delta_{ij} H_i, \quad  \{Q_i, Q_j\} = 0.
\end{equation}
Using the Jacobi identity, it is easy to prove that $H_1$ and~$H_2$ must commute. This algebra acts, for example, on the Ramond sector of any \(\mathcal N=2\) superconformal field theory in two dimensions. 
The relevant closed subalgebra of the Ramond algebra is:
\begin{equation}
\begin{aligned}
(G^\pm_0)^2&=0\\
(\bar G^\pm_0)^2&=0\\
\{G^\pm_0,\bar G^\pm_0\}&=0\\
\{G^+_0,G^-_0\}&=L_0-c/24\\
\{\bar G^+_0,\bar G^-_0\}&=\bar L_0-c/24.
\end{aligned}
\end{equation}

Clearly, representations are now labeled by \emph{two} energies, $E_1$ and~$E_2$. The irreducible representations are of four types, according to whether or not each $E_i$ is zero; their dimensions are four, two, two and one. 

The construction we gave above still goes through and produces a \(\CP1\) family of \(\mathcal N=1\) algebras with supercharges~$Q_t$. However, the \(\CP1\) family of Hamiltonians is now nontrivial:
\begin{equation}
\{ Q_t, Q_t^\dagger \} = 2 H_t = 2 \left( |t_1|^2 H_1 + |t_2|^2 H_2 \right) .
\end{equation}
The zero-energy spectrum of $H_t$ is the same for almost all~$t$, consisting of states for which $E_1$ and~$E_2$ are both zero. (These are the genuinely short one-dimensional representations.) However, there are two special points in the \(\CP1\) moduli space that exhibit \emph{enhanced} vacuum degeneracy: the points $[0:1]$ and $[1:0]$, corresponding to our two original supercharges. The cohomology of~$Q_t$ thus jumps in total rank at these points. 

One should note that these jumps do not occur unless there are states for which one Laplacian is zero but the other is not. The condition that the spaces of vacuum states for the two Laplacians agree is strictly weaker than the condition that the two operators are identical. Nonetheless, by formal Hodge theory, it is enough to ensure that $Q_t$-cohomology is the same for all~$t$, and therefore that only square and singlet representations occur. As such, it is sufficient to establish the $\dif\dif{}\conj$-lemma.

In quantum mechanics, one should expect enhanced degeneracy to be an avatar of enhanced symmetry. The obvious question is: what symmetry is enhanced at the points of our \(\CP1\) moduli space that exhibit extra BPS states? 

The answer is that precisely at these points \(\mathcal N=1\) supersymmetry is enhanced to \(\mathcal N=1^2\). It is easy to check that, when we try to define a second supercharge $Q_{\smash{\tilde{t}}}$ by our prescription above, something goes wrong: namely, 
\begin{equation*}
\{ Q_t, {Q_{\smash{\tilde{t}}}}^\dagger \} = t_1t_2 \left( H_2 - H_1 \right) \neq 0.
\end{equation*}
As such, when~$t$ is generic, no supercharges other than $Q_t$ and its adjoint can be found that act in a compatible way.

This leads us to the second interesting variant, which we will refer to as \(\mathcal N=1.5\). The notation is meant to convey not only intermediacy between \(\mathcal N=1\) and \(\mathcal N=2\), but also that something is ``not whole'': the commutation relations in this case no longer define a closed algebra. 
In this case, the following commutation relations are imposed:
\begin{equation}
\{ Q_i, Q_i^\dagger\} = 2 H_i, \quad  \{ Q_i, Q_j \} = 0.
\end{equation}
However, we do not require that $S=\{ Q_1, Q_2^\dagger \}$ must vanish. It is simply a new bosonic operator that can be defined, about which nothing can \emph{a priori} be said.

On a generic complex manifold, this is the algebra that holds between the derivatives \(\partial\) and \(\bar\partial\). Moreover, in a generic double complex, this is the algebra satisfied by the two differentials. We will return to this point when reviewing results from complex geometry in the next section.

New, exotic irreducible representations of this algebra are possible. In addition to the standard long and short representations of \(\mathcal N=2\) (``squares'' and ``dots''), one can also have ``staircases'' of the sort depicted below:\footnote{We are grateful to David Speyer for helpful comments at \url{http://mathoverflow.net/questions/86947/}
that pointed these facts out to us.}
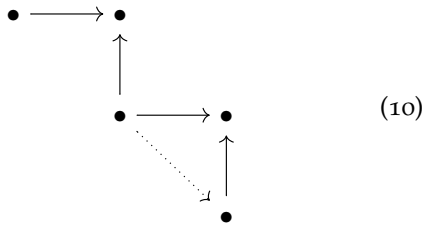
\begin{equation}
\begin{tikzcd}
\bullet \rar & \bullet &\\
& \bullet \rar\uar \drar[dotted] & \bullet \\
&& \bullet \uar
\end{tikzcd}
\label{staircase}
\end{equation}
As in~\protect\eqref{square}, the plane is the $(F_1, F_2)$-plane, and the arrows represent the action of the supercharges.  We will refer to the total dimension of such a representation as the ``length'' of the staircase, which is five for the staircase pictured above; any length can occur. The dashed line indicates the operator $S$ defined above, whose role is to go up and down stairs. 

Staircases of odd length, like the one pictured in~\protect\eqref{staircase}, contribute one generator to both $Q_1$-cohomology and $Q_2$-cohomology, as well as to the ``total'' cohomology of the supercharge $Q_1+Q_2$. However, as can easily be seen from the picture, the generator is in a different $(F_1-F_2)$-degree. We can no longer label cohomology classes with their ``axial'' quantum numbers. 

Furthermore, staircases of even length (depending on their orientation) contribute \emph{two} generators to one of the $Q_i$-cohomologies, but \emph{none} to the other. Drawing the appropriate picture will make this clear. These ideas should be familiar to anyone who is familiar with the spectral sequences associated to a double complex: they begin at $Q_1$- or $Q_2$-cohomology, converge to $(Q_1+Q_2)$-cohomology, and the differentials on the $k$-th page cancel pairs of generators that lie at opposite ends of a staircase of length $2k$. 

As always, the two anticommuting supercharges allow one to define a \(\CP1\) family of de~Rham algebras. However, unlike in the \(\mathcal N=1^2\) case, the Hamiltonians corresponding to different points in this moduli space can never be simultaneously diagonalized. The spectral sequence is a formal way of describing how different representations appear and disappear in the zero-energy spectrum of this parameterized family of Hamiltonians. 
Further remarks on the physics of spectral sequences and supersymmetric quantum mechanics will appear in~\cite{GNSSS-Forthcoming}.

\section{Review of K\"ahler, Calabi-Yau, and hyper-K\"ahler geometry}

We now give a brief review of some classical facts about K\"ahler geometry, 
in a way that is tailored to our purposes. For more details, the reader is referred to the excellent exposition in~\cite{DGMS}, or to the book~\cite{Huybrechts}. We will return to the formality result of~\cite{DGMS} when reviewing Sullivan's theory of minimal models in \protect\cref{sullivan}. Calabi-Yau and hyper-K\"ahler geometry provide levels of additional structure; we address these in order of decreasing generality. Readers should keep in mind that, in the context of two-dimensional superconformal sigma models, Calabi-Yau structure is required to define an \(\mathcal N=(2,2)\) SCFT. A hyper-K\"ahler target is necessary and sufficient for \(\mathcal N=(4,4)\) superconformal symmetry.

Let~\(M\) be a smooth manifold of real dimension~$2n$. An almost-complex structure on~\(M\) is a vector-bundle morphism $J\colon TM \to TM$ such that $J^2 = -1$. On the space of complex-valued differential forms, $J$ will induce a decomposition
\begin{equation}
\mathrm\Omega^n(M,\mathbb C) = \bigoplus_{p+q=n}\mathrm\Omega^{p,q}(M),
\end{equation}
which comes from the decomposition of the complexified tangent bundle into the $\pm i$ eigenspaces of~$J$.   
This decomposition allows us to write the exterior derivative operator $d$ as a sum of terms of different degree: 
\begin{equation}
d = \sum_{r+s=1} d^{r,s}.
\end{equation}
The almost-complex structure is said to be integrable when \(M\) is a complex manifold: it admits an atlas of complex-valued coordinates with holomorphic transition maps, such that multiplication by~$i$ agrees with the almost-complex structure~$J$. This occurs precisely when only terms of degree $(1,0)$ and~$(0,1)$ are present in the decomposition of~$d$. (An equivalent condition is that $d^{0,1}$ square to zero.) In this case, we use the symbols \(\partial\) and~\(\bar\partial\) for the operators $d^{1,0}$ and $d^{0,1}$; it is easy to check that these operators separately square to zero and anticommute with one another. A complex manifold is a space locally modeled on~$\mathbb C^n$; one is able to define which functions are holomorphic. 

The above shows immediately that the de~Rham complex of a complex manifold admits an action of \(\mathcal N=1.5\) supersymmetry. The de~Rham algebra (or \(\mathcal N=1\) supersymmetry algebra for quantum mechanics) acts in three meaningfully different ways:\footnote{As we emphasized in the previous section, a \(\CP1\) family of de~Rham algebras can be defined. The three actions we mention here are three representative points in this moduli space.} each of \(\partial\), \(\bar\partial\), and $d = \partial + \bar\partial$ may be thought of as a nilpotent supercharge. 

However, on a garden-variety complex manifold, supersymmetry is \emph{not} enhanced to \(\mathcal N=2\). While \(\partial\) anticommutes with $\bar \partial$, it may not commute with its adjoint: $\{ \partial, \bar\partial^\dagger \} \neq 0$. When this is the case (as we discussed in a general setting above) the respective Laplacians will not agree and will fail to commute with one another. Precisely when the manifold is K\"ahler, its de~Rham complex provides an \(\mathcal N=2\) supersymmetric quantum mechanics. We will return to this point after giving the formal definition of K\"ahler structure.

Suppose that~\(M\) is a complex manifold. A symplectic structure on~\(M\) is a choice of closed, non-degenerate two-form~$\omega \in \mathrm\Omega^2(M,\mathbb R)$. One says that the symplectic structure~$\omega$ is \emph{compatible} with the complex structure if the composite tensor
\begin{equation}
g(a,b) = \omega(a, Jb) 
\end{equation}
is a Riemannian metric on~\(M\). Symmetry of the metric tensor then implies that the form~$\omega$ has $(p,q)$-degree $(1,1)$.  A complex manifold admitting such a compatible triple of structures is called \emph{K\"ahler}. 

None of the conditions in the definition of a K\"ahler manifold can be relaxed; counterexamples exist in all cases. For example, there are complex manifolds that admit symplectic forms, but nonetheless are not K\"ahler (the compatibility condition cannot be satisfied). Despite these rigid requirements, though, many manifolds are naturally K\"ahler. All of the complex projective spaces~$\CP n$ are, when equipped with the Fubini--Study metric. Furthermore, a complex submanifold of a K\"ahler manifold is again K\"ahler; as such, so are all smooth projective algebraic varieties over~\(\mathbb C\). 

Just as the integrability condition on an almost-complex structure may be phrased in terms of differential operators as the identity $\bar\partial^2=0$, so the conditions for \(M\) to be K\"ahler can be expressed by a set of \emph{K\"ahler identities} between operators in its de~Rham complex. 

Define an operator $\Lambda^\dagger$ by the relation
\begin{equation*}
\Lambda^\dagger\colon \alpha \mapsto \omega \wedge \alpha,
\end{equation*}
where~$\omega$ is the K\"ahler form, and let~$\Lambda = \star \Lambda^\dagger \star$ be its adjoint.
This pair of operators are called \emph{Lefschetz operators;} 
they have pure degree $(1,1)$ and $(-1,-1)$ respectively. Their commutator $D = [\Lambda^\dagger, \Lambda]$ is a degree operator: it acts diagonally according to the rule
\begin{equation*} 
D|_{\mathrm\Omega^{p,q}} = (p+q-n).
\end{equation*}
To check that this is true, one can simply do a linear algebra calculation corresponding to one of the fibers of the bundle. For a Hermitian vector space that is one-dimensional over~\(\mathbb C\), the result is easy to see; the exterior algebra $\Lambda^* V$ is concentrated in degrees 0, 1, and~2, so that $[\Lambda^\dagger, \Lambda]$ must be zero in the middle dimension. The general statement is then obtained by induction on the complex dimension of~$V$; for details, the reader is referred to~\cite{Huybrechts}.

It follows that these three operators together form a basis for the Lie algebra $\SL(2,\mathbb R)$: they satisfy the commutation relations
\begin{equation*}
[D,\Lambda^\dagger] = 2 \Lambda^\dagger, \quad
[D, \Lambda] = -2 \Lambda, \quad
[\Lambda^\dagger, \Lambda ] = D. 
\end{equation*}
(Since the K\"ahler form is real, these operators actually act on the real cohomology and not merely on cohomology with complex coefficients; however, the latter is the relevant case for quantum mechanics.) The K\"ahler identities imply that K\"ahler manifolds admit a unique choice of Laplacian, that their de~Rham complexes are examples of \(\mathcal N=2\) supersymmetric quantum mechanics, and that the $\SL(2,\mathbb R)$ action defined above commutes with the Laplacian and hence descends to the cohomology. The essential identities, from which all others can be derived, are
\begin{gather}
[\Lambda, \bar\partial] = - i \partial^\dagger,
\quad
[\Lambda, \partial] = i \bar\partial^\dagger,
\label{Kahler1}
\\ \nonumber
[\Lambda, \partial^\dagger] = [\Lambda, \bar\partial^\dagger] = 0.
\end{gather}
To establish these, we refer to~\cite{DGMS}, who point out that K\"ahler metrics osculate to second order to the flat metric on~$\mathbb C^n$. This implies that any identity which is at most first-order in derivatives of the metric can be established by checking it on~$\mathbb C^n$. The calculation in the case of~\protect\eqref{Kahler1} is simple. 

Given these identities, it is straightforward to show that supersymmetry is enhanced:
\begin{align}
\{ \partial, \bar\partial^\dagger \} &= i \{ \partial, [\partial, \Lambda] \} \nonumber \\
&= i \left( \partial \partial \Lambda - \partial \Lambda \partial + \partial \Lambda \partial - \Lambda \partial \partial \right) \nonumber \\
&= 0 \quad (\text{since $\partial^2=0$}). 
\label{Kahler2}
\end{align}
Therefore, two mutually commuting copies of the \(\mathcal N=1\) algebra act. It remains to show that their respective Laplacians agree. This is again a quick calculation using~\protect\eqref{Kahler1}, the Jacobi identity, and $\{\partial, \bar\partial\} = 0$:
\begin{align}
\{ \partial, \partial^\dagger \} &=  i \{ \partial, [\Lambda, \bar\partial ] \} \nonumber \\
&= i \left(  [\Lambda, \{ \partial, \bar\partial \} ] +
\{ \bar\partial, [\partial, \Lambda] \} \right) \nonumber \\
&= \{ \bar\partial, i [ \partial, \Lambda ] \} \nonumber \\
&= \{ \bar\partial, \bar\partial^\dagger \} . 
\label{Kahler3}
\end{align}
As a consequence of~\protect\eqref{Kahler3}, there is a unique harmonic representative of each cohomology class, with pure $(p,q)$-degree. 
Recall that the representations of \(\mathcal N=2\) supersymmetry are of two types: zero-energy, one-dimensional ``short'' representations, and positive-energy, four-dimensional representations, which form ``squares'' in degrees $(p,q)$, $(p+1,q)$, $(p,q+1)$, and~$(p+1,q+1)$. 
The Betti numbers of the manifold can therefore be refined into Hodge numbers,
\begin{equation}
b_n = \sum_{p+q=n} h_{p,q},
\label{HodgeBetti}
\end{equation} 
which count the short multiplets. (We emphasize once again that this is not true on a generic complex manifold! \protect\cref{HodgeBetti} is a consequence of the degeneration of the Hodge-to-de-Rham spectral sequence, which is in turn a consequence of the enhanced supersymmetry guaranteed by the K\"ahler identities.)

One usually arranges the Hodge numbers of~\(M\) in a diamond, according to their bidegrees, as shown in \protect\cref{HodgeDiamond}.

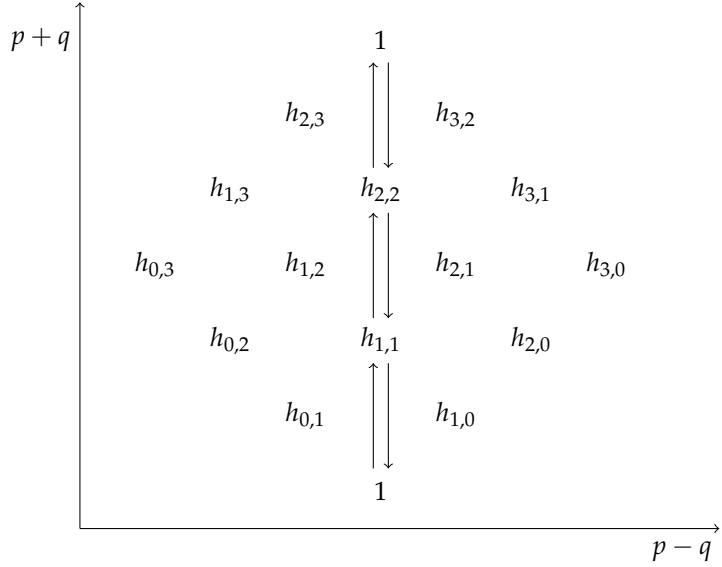
\begin{figure*}
\begin{center}
\begin{tikzpicture}[x=1cm,y=1cm]
\draw (0,0) node {\(1\)};
\draw (-1,1) node {$h_{0,1}$};
\draw (1,1) node {$h_{1,0}$};
\draw (-2,2) node {$h_{0,2}$};
\draw (0,2) node {$h_{1,1}$};
\draw (2,2) node {$h_{2,0}$};
\draw (-3,3) node{$h_{0,3}$};
\draw (-1,3) node{$h_{1,2}$};
\draw (1,3) node{$h_{2,1}$};
\draw (3,3) node{$h_{3,0}$};
\draw (-2,4) node {$h_{1,3}$};
\draw (0,4) node {$h_{2,2}$};
\draw (2,4) node {$h_{3,1}$};
\draw (-1,5) node {$h_{2,3}$};
\draw (1,5) node {$h_{3,2}$};
\draw (0,6) node {\(1\)};

\draw[->] (-4,-0.5) -- (-4,6) node[left] {$p+q$} -- (-4,6.5);
\draw[->] (-4,-0.5) -- (4,-0.5) node[below] {$p-q$} -- (4.5,-0.5);

\draw[->] (-0.1,2.3) -- (-0.1,3.7);
\draw[->] (0.1,3.7) -- (0.1,2.3);

\draw[->] (-0.1,0.3) -- (-0.1,1.7);
\draw[->] (0.1,1.7) -- (0.1,0.3);

\draw[->] (-0.1,4.3) -- (-0.1,5.7);
\draw[->] (0.1,5.7) -- (0.1,4.3);

\end{tikzpicture}
\end{center}
\caption[The Hodge diamond of a connected K\"ahler threefold]{The Hodge diamond for a generic connected K\"ahler threefold. The raising and lowering operators, indicated by vertical arrows for $p-q=0$, are the Lefschetz operators. Many of the indicated Hodge numbers are not independent, being related by discrete symmetries.}
\label{HodgeDiamond}
\end{figure*}

The K\"ahler identities~\protect\eqref{Kahler1} further imply that the Lefschetz operators commute with the Laplacian; checking this is easy, and proceeds using the Jacobi identity. The action of $\SL(2,\mathbb R)$ they provide therefore maps harmonic forms to harmonic forms, and so descends to the cohomology.

As such, a nice way to think of the Hodge diamond of a K\"ahler manifold is as a weight diagram for a representation of the rank-two Lie algebra $\SU(2) \times \operatorname U(1)$. The algebra acts via the operators $\Lambda^\dagger$, $\Lambda$, $D=P+Q-n$, and~$P-Q$. The Hodge numbers are then arranged by their weights with respect to the Cartan of this algebra. This picture will generalize nicely to the hyper-K\"ahler case.

The Lefschetz operators have extra consequences for the topology of~\(M\), in addition to those that follow from enhanced supersymmetry. Since they commute with $P-Q$, it follows immediately that each vertical slice of the Hodge diamond (when $(p-q)$-degree is displayed horizontally) is a weight diagram for $\SL(2,\mathbb R)$. As a consequence, the odd- and even-degree Betti numbers of a K\"ahler manifold are separately monotonically nondecreasing toward the middle degree. This is the content of the hard Lefschetz theorem, which is straightforward to understand from the standpoint of representation theory.

While \(\mathcal N=(2,2)\) superconformal symmetry in two dimensions is enough to imply things like bigrading and Poincar\'e duality---we return to these points later---it is not sufficient to establish properties like hard Lefschetz. Indeed, $\mathcal N =(2,2)$ theories may not even have operators with the appropriate quantum numbers to correspond to a K\"ahler class---superconformal minimal models are examples of this.

Since the volume form of a compact K\"ahler manifold has bidegree $(n,n)$, Poincar\'e duality defines a pairing between $\mathrm H^{p,q}$ and $\mathrm H^{n-p,n-q}$, implying immediately that the corresponding Hodge numbers are equal. Poincar\'e duality acts on the Hodge diamond~(\protect\cref{HodgeDiamond}) by reflection through the center point, or equivalently by simultaneously flipping the sign of $P-Q$  and $D$.

An additional symmetry of the Hodge diamond comes from considering the action of complex conjugation on the de~Rham complex. This takes forms of degree~$(p,q)$ to forms of degree~$(q,p)$, while preserving total degree.  It therefore acts on the Hodge diamond by reflecting left and right, flipping the sign of~$P-Q$ while fixing~$D$.

From the representation-theory standpoint, we are considering unitary representations of \(\SU(2)\times\operatorname U(1)\). Every representation of \(\SU(2)\) is 
symmetric with respect to the Weyl group $\mathbb Z/2\mathbb Z$, which acts on the Hodge diamond by reflections about the horizontal axis. 
Furthermore, for unitary representations of U(1), we also insist on charge-conjugation symmetry, which  reflects the Hodge diamond about the vertical axis. We therefore recover the same $(\mathbb Z/2\mathbb Z)^2$ symmetry that is generated by Poincar\'e duality and complex conjugation.

Lastly, the so-called Hodge-Riemann bilinear relation is a compatibility condition 
between the Lefschetz action and these discrete symmetries: it states that, for highest-weight states, the Poincar\'e dual of the complex conjugate is the same (up to a complex scalar phase) as the state obtained by applying the \(\SU(2)\) raising operator multiple times. In pictures, the following should commute up to a scalar:
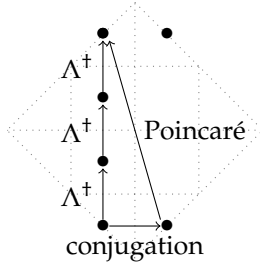
\begin{center}
\begin{tikzpicture}[scale=0.85]
\drawdiamond2
\node at (0.5,1.5) {\(\bullet\)};
\node at (0.5,-1.5) {\(\bullet\)};
\node at (-0.5,1.5) {\(\bullet\)};
\node at (-0.5,-1.5) {\(\bullet\)};
\draw [->] (0.5,-1.5) ++(-0.1,0.1) -- (0,0) node [right] {Poincar\'e} -- ({-0.5+0.1},{1.5-0.1});
\draw [->] (-0.5,-1.5)  ++(0.1,0) -- (0,-1.5) node [below] {conjugation} -- ({0.5-0.1},-1.5);

\node at (-0.5,-0.5) {\(\bullet\)};
\node at (-0.5,0.5) {\(\bullet\)};
\node at (-0.5,1.5) {\(\bullet\)};
\draw [->] (-0.5,-1.5) ++ (0,0.1) -- (-0.5,-1) node [left] {\(\Lambda^\dagger\)} -- (-0.5,{-0.5-0.1});
\draw [->] (-0.5,-0.5) ++ (0,0.1) -- (-0.5,0) node [left] {\(\Lambda^\dagger\)} -- (-0.5,{0.5-0.1});
\draw [->] (-0.5,0.5) ++ (0,0.1) -- (-0.5,1) node [left] {\(\Lambda^\dagger\)} -- (-0.5,{1.5-0.1}); 
\end{tikzpicture}
\end{center}
The equation that expresses this is 
\begin{equation}
\star (\Lambda^\dagger)^j \psi = \frac{(-1)^{k(k+1)/2} j!}{(n-k-j)!} (\Lambda^\dagger)^{n-k-j} \bar\psi,
\end{equation}
where~$\psi$ is a $k$-form such that $\Lambda\psi=0$. 
See~\cite{Huybrechts} for more details, as well as a proof of the K\"ahler identities that starts from this formula.

A K\"ahler manifold \(M\) is \emph{Calabi-Yau} if it has a nowhere-vanishing holomorphic---and therefore harmonic---$(n,0)$-form. (An equivalent condition is that the canonical bundle be trivial.) Therefore, a Calabi-Yau manifold will have $h_{n,0} = h_{0,n} = 1$. 

There are strong constraints on the fundamental groups of K\"ahler and Calabi-Yau manifolds as well. Just to substantiate this, we recall that any Calabi-Yau~\(M\) has a finite cover of the form $T\times N$, where \(n\) is a simply connected Calabi-Yau and~$T$ is a complex torus. This fact is called \emph{Bogomolov decomposition}~\cite{Bogomolov-en}. In this paper, however, we will always restrict ourselves to the simply connected case.

The next level of possible structure is provided by \emph{hyper-K\"ahler manifolds}, which are Riemannian manifolds equipped with three distinct complex structures. We will call these $I$, $J$, and~$K$. They should satisfy the multiplication table of the quaternions: for instance, $IJ = - JI = K$. Further, when equipped with any of these complex structures, the manifold should be K\"ahler. The tangent spaces to a hyper-K\"ahler manifold are thus quaternionic vector spaces; while this is necessary, it is not sufficient. (Again, it is possible for compatibility conditions between the complex structures and the metric to fail; so-called quaternionic K\"ahler manifolds are examples of this~\cite{Dancer}.)

Hyper-K\"ahler manifolds are highly structured, and few compact examples are known: the hyper-K\"ahler structure imposes strong constraints on the topology. We review some of these constraints briefly here.

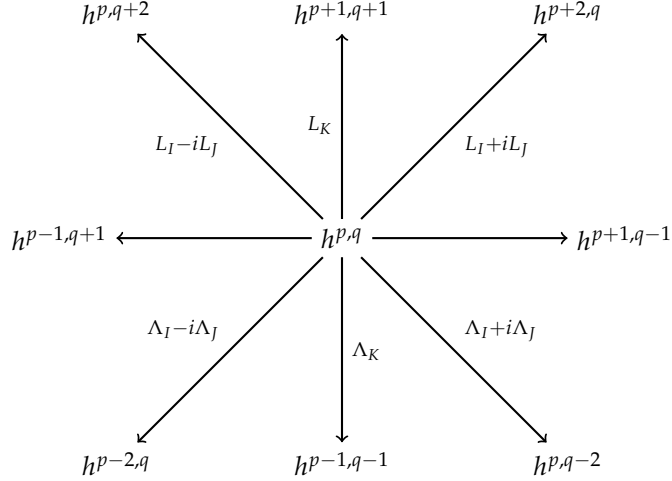
\begin{figure*}[t]
\begin{center}
\begin{tikzpicture}[x=3cm,y=3cm]

\node at (0,0) (O) {\(h^{p,q}\)};
\draw (O)++(1,1) node (O++) {\(h^{p+2,q}\)};
\draw (O)++(1,-1) node (O+-) {\(h^{p,q-2}\)};
\draw (O)++(-1,1) node (O-+) {\(h^{p,q+2}\)};
\draw (O)++(-1,-1) node (O--) {\(h^{p-2,q}\)};
\draw (O)++(0,1) node (O0+) {\(h^{p+1,q+1}\)};
\draw (O)++(0,-1) node (O0-) {\(h^{p-1,q-1}\)};
\draw (O)++(1,0) node[right] (O+0) {\(h^{p+1,q-1}\)};
\draw (O)++(-1,0) node[left] (O-0) {\(h^{p-1,q+1}\)};

\draw [thick,->] (O) -- ++(0.5,0.5) node[below right] {\(\scriptstyle L_I + i L_J \)} -- (O++);
\draw [thick,->] (O) -- ++(0.5,-0.5) node[above right] {\(\scriptstyle \Lambda_I + i \Lambda_J \)} -- (O+-);
\draw [thick,->] (O) -- ++(-0.5,-0.5) node[above left] {\(\scriptstyle \Lambda_I - i \Lambda_J \)} -- (O--);
\draw [thick,->] (O) -- ++(-0.5,0.5) node[below left] {\(\scriptstyle L_I - i L_J \)} -- (O-+);
\draw [thick, ->] (O) -- ++(0,0.5) node [left] {\(\scriptstyle L_K \)} -- (O0+);
\draw [thick, ->] (O) -- ++(0,-0.5) node [right] {\(\scriptstyle \Lambda_K \)} -- (O0-);
\draw [thick, ->] (O) -- (O+0);
\draw [thick, ->] (O) -- (O-0);
\end{tikzpicture}
\end{center}
\caption{The action of the \(B_2\) root system on the Hodge diamond of a hyper-K\"ahler manifold.}
\label{B2}
\end{figure*}

The three K\"ahler classes corresponding to the three complex structures define operators on the de~Rham complex which close into the action of a Lie algebra, analogous to the Lefschetz $\SL(2)$ action in the K\"ahler case. The algebra that applies in this case is $\SO(5)$; it was first constructed by Verbitsky~\cite{Verbitsky-en}. For the reader's convenience, we explain his construction here. 
We use the notation
\[ L_I = \Lambda_I^\dagger\colon \alpha \mapsto \omega_I \wedge \alpha \]
(and its analogues) for the Lefschetz-type operators. (One should be careful here; $\Lambda_I$ is the adjoint of~$L_I$ with respect to complex structure~$I$. However, if one works with respect to a fixed complex structure, not all pairs of Lefschetz operators will be adjoints; this is why we prefer the notation~$L$ for raising operators in the sequel.)

Following Verbitsky, let's define
\[ M_{IJ} = [L_I, \Lambda_J], \]
where the indices range over pairs of complex structures. Clearly,
\[ M_{II} = M_{JJ} = M_{KK} = D. \]
Moreover, $[L_I, L_J]=[\Lambda_I,\Lambda_J] = 0$ for any pair of indices $I$ and~$J$, since two-forms commute. The~\(M\) operators carry homological degree zero, and therefore commute with~$D$.
 
One identity requires a nontrivial argument: $M_{IJ} = - M_{JI}$. Just as in the K\"ahler case, commutation relations like this one are really statements about linear algebra, and can be proved by an explicit calculation in the case of a one-dimensional quaternionic vector space.

Once this is established,
the remaining commutation relations can be fixed by some quick calculations with the Jacobi identity:
\begin{gather}
[M_{IJ}, \Lambda_I] = - 2\Lambda_J , \quad [M_{IJ}, L_J] = 2 L_I; \nonumber \\
[M_{IJ}, \Lambda_K] = [M_{IJ}, L_K = 0]; \\
[M_{IJ}, M_{JK}] = 2 M_{IK}. \nonumber 
\end{gather}
These operators therefore close into a ten-dimensional Lie algebra of rank two. Let's choose to fix the complex structure $K$ on the manifold, and the corresponding Cartan subalgebra spanned by~$D$ and~$-iM_{IJ}$. It is then straightforward to show that the weights of the adjoint representation form the root system $B_2$, corresponding to the algebra $\SO(5)$. Moreover, the bigrading defined by weights for this choice of Cartan on the de~Rham complex coincides with the Hodge bigrading by homological degree and $(P-Q)$, defined with respect to complex structure~$K$. That is,
\[
-i M_{IJ} = (P-Q).
\]
(See \protect\cref{B2}.)

The picture of the Hodge diamond as a weight diagram for a Lie algebra representation furnished by the cohomology therefore generalizes beautifully to the hyper-K\"ahler case. The Lie algebra in question is now $\SO(5)$, rather than $\SU(2) \times\operatorname U(1)$, and the restrictions imposed on the Hodge numbers by representation theory are accordingly more severe. For example, only three irreducible representations of $\SO(5)$ can fit inside the Hodge diamond of a hyper-K\"ahler four-manifold, and only six inside that of an eight-manifold. (See \protect\cref{SO5-irreps}.) The restrictions coming from the~$\SO(5)$ action are best thought of as inequalities, rather than equalities, between Hodge numbers: they are the analogue of the monotonicity properties guaranteed for Betti numbers of K\"ahler manifolds by the hard Lefschetz theorem.

The discrete symmetries of the Hodge diamond are also enhanced in the hyper-K\"ahler case. The Weyl group of $B_2$ is $(\mathbb Z/2\mathbb Z)^3$; since this group has order eight, Hodge numbers are repeated up to eight times. The new identity is
\[
h^{p,q} = h^{p,2n-q} \quad (n = \dim_{\mathbb C}M).
\]

\begin{figure*}[t]
\centering
\begin{subfigure}{0.32\textwidth}\centering
\begin{tikzpicture}[scale=0.85]
\drawdiamond1
\node at (0,0) {\(\bullet\)};
\node at (0,-1) {\(\phantom{\bullet}\)};
\node at (0,1) {\(\phantom{\bullet}\)};
\end{tikzpicture}
\caption{\(\mathbf1\)}
\end{subfigure}\begin{subfigure}{0.32\textwidth}\centering
\begin{tikzpicture}[scale=0.85]
\drawdiamond1
\node at (0.5,0.5) {\(\bullet\)}; 
\node at (-0.5,0.5) {\(\bullet\)};
\node at (0.5,-0.5) {\(\bullet\)};
\node at (-0.5,-0.5) {\(\bullet\)};
\node at (0,-1) {\(\phantom{\bullet}\)};
\node at (0,1) {\(\phantom{\bullet}\)};
\end{tikzpicture}
\caption{\(\mathbf4\)}
\end{subfigure}\begin{subfigure}{0.32\textwidth}\centering
\begin{tikzpicture}[scale=0.85]
\drawdiamond1
\node at (0,0) {\(\bullet\)}; 
\node at (1,0) {\(\bullet\)};
\node at (-1,0) {\(\bullet\)};
\node at (0,-1) {\(\bullet\)};
\node at (0,1) {\(\bullet\)};
\end{tikzpicture}
\caption{\(\mathbf5\)}
\end{subfigure}

\begin{subfigure}{0.32\textwidth}\centering
\begin{tikzpicture}[scale=0.85]
\drawdiamond2
\node at (0,0) {\(\bullet\bullet\)};
\node at (1,0) {\(\bullet\)};
\node at (-1,0) {\(\bullet\)};
\node at (0,-1) {\(\bullet\)};
\node at (0,1) {\(\bullet\)};
\node at (1,1) {\(\bullet\)}; 
\node at (1,-1) {\(\bullet\)};
\node at (-1,1) {\(\bullet\)};
\node at (-1,-1) {\(\bullet\)};
\end{tikzpicture}
\caption{\(\mathbf{10}\)}
\end{subfigure}
\begin{subfigure}{0.32\linewidth}\centering
\begin{tikzpicture}[scale=0.85]
\drawdiamond2

\node at (0,0) {\(\bullet\bullet\)};
\node at (1,0) {\(\bullet\)};
\node at (-1,0) {\(\bullet\)};
\node at (0,-1) {\(\bullet\)};
\node at (0,1) {\(\bullet\)};
\node at (1,1) {\(\bullet\)}; 
\node at (1,-1) {\(\bullet\)};
\node at (-1,1) {\(\bullet\)};
\node at (-1,-1) {\(\bullet\)};
\node at (2,0) {\(\bullet\)};
\node at (-2,0) {\(\bullet\)};
\node at (0,2) {\(\bullet\)};
\node at (0,-2) {\(\bullet\)};
\end{tikzpicture}
\caption{\(\mathbf{14}\)}
\end{subfigure}\begin{subfigure}{0.32\linewidth}\centering
\begin{tikzpicture}[scale=0.85]
\drawdiamond2

\node at (0.5,0.5) {\(\bullet\bullet\)}; 
\node at (-0.5,0.5) {\(\bullet\bullet\)};
\node at (0.5,-0.5) {\(\bullet\bullet\)};
\node at (-0.5,-0.5) {\(\bullet\bullet\)};
\node at (1.5,0.5) {\(\bullet\)};
\node at (1.5,-0.5) {\(\bullet\)};
\node at (-1.5,0.5) {\(\bullet\)};
\node at (-1.5,-0.5) {\(\bullet\)};
\node at (0.5,1.5) {\(\bullet\)};
\node at (0.5,-1.5) {\(\bullet\)};
\node at (-0.5,1.5) {\(\bullet\)};
\node at (-0.5,-1.5) {\(\bullet\)};
\node at (0,-2) {\(\phantom{\bullet}\)};
\node at (0,2) {\(\phantom{\bullet}\)};
\end{tikzpicture}
\caption{\(\mathbf{16}\)}
\end{subfigure}
\caption[Some weight diagrams for representations of~$\SO(5)$]{\(\SO(5)\) irreps that can fit inside an 8-dimensional Hodge diamond. The first three can also fit inside a 4-dimensional Hodge diamond. The $\mathbf{16}$ representation cannot occur for simply connected eight-manifolds.}\label{SO5-irreps}
\end{figure*}
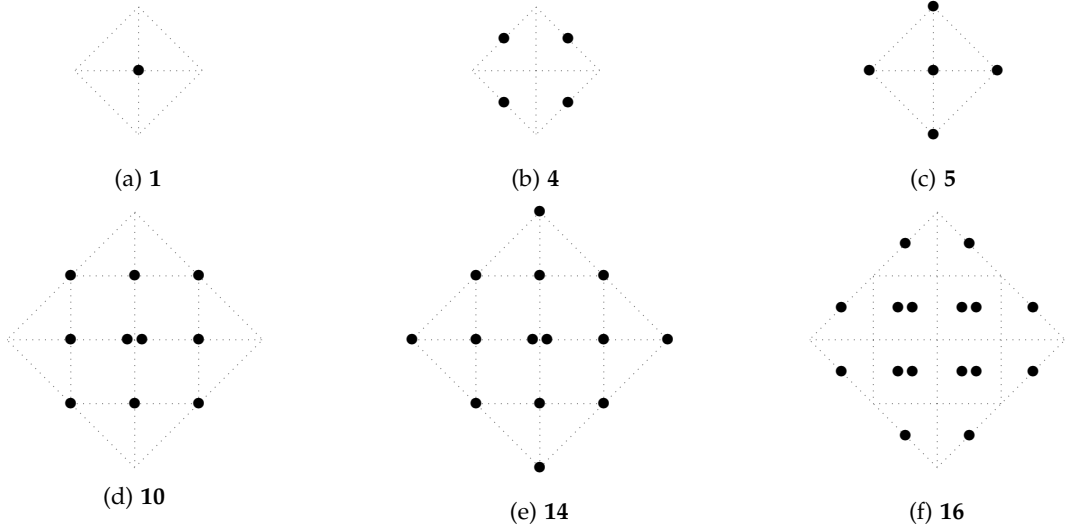

Moreover, the de~Rham complex on a hyper-K\"ahler manifold has differential operators generating an action of the $\mathcal N=4$ supersymmetry algebra. Recall that, on a K\"ahler manifold, the real operators were the exterior derivative~$d$ and its partner $d^c = -JdJ$. On a hyper-K\"ahler manifold, we can write the four operators
\begin{equation*}
d_1 = d, \quad
d_I = -IdI, \quad 
d_J = -JdJ, \quad 
d_K = -KdK.
\end{equation*}
The commutation relations for these operators and their adjoints follow trivially from those for~$d$ and~$d^c$, since for any pair of them one can be obtained from the other by conjugating with a specific complex structure.

To emphasize the perspective we have tried to bring out in this review: the requirements on the target space of supersymmetric sigma models, which are usually understood in terms of special holonomy, can also fruitfully be thought about in terms of dimensional reduction to supersymmetric quantum mechanics. Since dimensional reduction preserves the number of real supercharges, a minimally supersymmetric \(d\)-dimensional sigma model can be defined on a target manifold~\(M\) if and only if the appropriate number of  supercharges ($2^{\lfloor d/2 \rfloor}$) act on the corresponding supersymmetric quantum mechanics (furnished by the de~Rham complex of~\(M\)). Although this condition is obviously necessary, it is perhaps surprising that it is sufficient. Furthermore, as was first pointed out by~\cite{FOFKS}, the perspective of dimensional reduction offers a natural explanation of Lefschetz-type symmetry algebras, just as it does for $R$-symmetries in general supersymmetric theories. In fact, one should think of Lefschetz symmetries in geometry and $R$-symmetries in supersymmetry as the same kind of object. We return to this in Section~\ref{fofks}.

\section{Homotopy theory over a field and Sullivan minimal models}
\label{sullivan}

In this section, we introduce readers to some rudiments of rational homotopy theory, as developed by Sullivan and collaborators, which formalizes  the idea of the de~Rham complex as an algebraic model of a manifold. We closely follow the exposition in~\cite{DGMS}; other more exhaustive references are~\cite{Sullivan-Infinitesimal,Griffiths-Morgan}. We attempt to use quantum-mechanical language when possible, and offer a reinterpretation of Sullivan's minimal models  in physical language, which sheds intuitive light on the origin of Massey products.

The de~Rham complex is a commutative differential graded algebra (briefly, a CDGA). That is, it is a graded algebra $\Omega = \bigoplus_{i\ge 0} \Omega^i$ over a ground field~${k}$ of characteristic zero, 
whose multiplication preserves the grading, equipped with a differential of degree one that is a derivation for the product:
\begin{equation}
d(x\cdot y) = dx\cdot y + (-)^{|x|} x \cdot dy.
\end{equation}
Here $x$ and~$y$ are assumed to have homogeneous degree, and $|x| = \deg x$. 
The product is further taken to be commutative in the graded sense familiar from supersymmetry:
\begin{equation}
x\cdot y = (-)^{|x||y|} y\cdot x.
\end{equation}
In words, a CDGA is a chain complex equipped with a compatible  notion of multiplication.

The cohomology of a CDGA is defined as usual, and is a commutative graded algebra. It is itself a CDGA, if we understand its differential to be the zero map. We will say that $\Omega$ is \emph{connected} if $\mathrm H^0(\Omega) = k$, and \emph{simply connected} if it is connected and $\mathrm H^1(\Omega) = 0$.

Homomorphisms of CDGAs are algebra homomorphisms that are also morphisms of chain complexes. A CDGA homomorphism is a \emph{quasi-isomorphism} if the map it induces on cohomology is an isomorphism. We will consider two CDGAs to be \emph{equivalent} if there is a zigzag of quasi-isomorphisms going from one to the other.

There is a collection of ideals $I_j$ in $\Omega$, defined by
\begin{equation}
I_j = \bigoplus_{i\ge j} \Omega^i \subset \Omega. 
\end{equation}
The ideal of \emph{decomposable elements} is $I_1\cdot I_1 \subseteq I_2$. 
It consists of elements of~$\Omega$ that can be written as products of elements of strictly lower degree. 

Let \(M\) be a simply connected CDGA. We will say that \(M\) is \emph{minimal} if (1) as an algebra, it is the free (graded)-commutative algebra generated by a finite number of elements of homogeneous degree; (2) the differential is decomposable, i.e.,
\[ \im d \subseteq I_1\cdot I_1.\]
From these conditions, it follows that $M^0 = k$ and $M^1 = 0$. The notion of a minimal CDGA can be generalized to the non-simply connected case, but for simplicity we will always deal with simply connected spaces and CDGAs in this paper. 

Like any CDGA, a minimal CDGA~\(M\) over~\(\mathbb C\) can be thought of as a supersymmetric quantum mechanics. If we choose an inner product on~\(M\) (or if one occurs naturally), so that the adjoints of operators can be defined, there is an obvious action of the relevant \(\mathcal N=1\) supersymmetry algebra. The requirement that~\(M\) be freely generated as a commutative algebra just says that it is a Fock space: each state can be built up from the (unique) vacuum state in degree zero by the free action of a finite number of bosonic and fermionic creation operators. This is what we would expect for the Hilbert space of a free field theory, even in $0+1$ dimensions. Moreover, once states are identified with operators in this way, the algebra structure is just the obvious product of operators. Structures analogous to this one can be seen (for instance) in the description of the elliptic genus in two-dimensional $\mathcal N=(2,2)$ theories in terms of a plethystic exponent.

Lastly, the condition that the differential be decomposable simply insists that we include no ``irrelevant'' creation operators that are $Q$-exact and so contribute only to the excited spectrum of the theory.

A minimal CDGA is a \emph{minimal model} of~$\Omega$ if there is a quasi-isomorphism 
\[ f\colon M \to \Omega. \]
The utility of minimal models comes from a theorem of Sullivan, which shows that every simply connected CDGA has a minimal model that is unique up to isomorphism. Furthermore, two CDGAs are equivalent if and only if their minimal models are isomorphic.
As such, minimal models cure the ``ambiguity'' of the de~Rham complex, and furnish a true invariant of a space. It is as though one had a way of choosing a unique triangulation (and therefore a unique simplicial chain complex) representing each topological manifold.

A minimal model of a CDGA can be thought of as the unique ``smallest'' supersymmetric quantum mechanics that (1) reproduces the correct BPS spectrum and (2) has the Hilbert space of a free theory. 
One can construct a minimal model degree by degree, adding generators only as necessary to generate or kill new cohomology. We will give examples below.

Given a minimal model~\(M\), we can consider its spaces of \emph{indecomposable elements} in each degree:
\[
\pi_i \doteq M^i/(I_1 \cdot I_1)^i.
\]
These vector spaces (more precisely, their duals) are the \emph{$k$-de~Rham homotopy groups} of~\(M\). By another remarkable theorem, if~\(M\) is the minimal model of the de~Rham complex of a space~\(X\) (with coefficients in~$k$), then the de~Rham homotopy groups of~\(M\) \emph{are} the homotopy groups of~\(X\), up to information about torsion:
\begin{equation}
\pi_i(M) \cong \pi_i(X) \otimes k.
\end{equation}
In fact, the minimal model contains all of the information of the homotopy type of~\(X\) over the field~$k$; if $k=\mathbb R$, in the words of~\cite{DGMS}, \(M\) ``\emph{is} the real homotopy type'' of~\(X\). While it does not preserve the complete homotopy type of~$X$ (since there is no information about torsion), it still provides a powerful algebraic invariant.

\begin{example}
The cohomology of the sphere $S^n$ is a simply connected CDGA with zero differential and zero product, with generators in degrees $0$ and~\(n\) only ($n>1$). If \(n\) is odd, this CDGA is already minimal, and so is its own minimal model.

If \(n\) is even, the minimal model \(M\) must have a free bosonic generator (call it~$x$) in degree~\(n\). The generator \(x\) must be closed in order to survive to the cohomology. Therefore \(x^2\) is closed; however, since it does \emph{not} survive to cohomology, it must be exact. We are therefore forced to introduce a generator~\(y\) in degree \(2n-1\), such that \(\dif y=x^2\).  This kills all higher powers of~$x$ (\(\dif{}(yx^n) = x^{n+2}\)), and therefore (since \(y\) is fermionic, so that \(y^2=0\)) we have constructed the minimal model.
\end{example}

It is \emph{a priori} not obvious that these minimal models carry any information about spheres as spaces; they are minimal models of the cohomology, rather than of the de~Rham complexes, of the spheres. However, it is possible to show that spheres are formal spaces, in the sense discussed in the introduction. To recall: a minimal CDGA~\(M\) is \emph{formal} if there is a CDGA quasi-isomorphism
\[ f\colon M \to\mathrm H^\bullet(M). \]
The algebras we constructed for spheres are therefore formal by definition; to show that spheres are formal spaces requires showing that their de~Rham complexes are equivalent to their cohomology. 

Once this is done, however, the minimal models we constructed above \emph{are} the minimal models of the \(n\)-spheres, and their de~Rham homotopy groups are the homotopy groups of spheres, after tensoring with~\(\mathbb R\). This shows that all the higher homotopy of odd-dimensional spheres is torsion, whereas the higher homotopy of even-dimensional spheres has rank one in degree $2n-1$ (corresponding for \(n=2\) to the Hopf fibration). Formality is thus a remarkably powerful tool.

\begin{example} 
We now discuss a minimal CDGA that fails to be formal. The example is, once again, due to~\cite{DGMS}. Let~\(M\) be freely generated by \(x\) and~\(y\) in degree two and \(f\) and~\(g\) in degree three. The differential will be defined by 
\begin{align}
\dif x = \dif y &= 0 &
\dif f &= x^2 &
\dif g &= xy.
\end{align}

This algebra cannot be formal, because it has a triple Massey product. The product $xxy$ is zero in cohomology for two reasons (because $xx$ and~$xy$ are both exact), and the Massey product measures the ``difference'' between the two. It is defined by the element 
\[ m(x,x,y) =  fy - gx, \]
which is closed but not exact in degree five. The Massey product is ambiguous up to the ideal
\[ x \cdot \mathrm H^3(M) + y \cdot \mathrm H^3(M),\] but this vanishes since \(\mathrm H^3(M)=0\).

\end{example}

\begin{figure*}

\newcommand\downpair[4]{
  \node at (#1,-#2+1) {\(#3\vphantom{y^2}\)};
  \node at (#1,{-#2}){\(#4\vphantom{y^2}\)};
  \draw [->] (#1,-#2+1) ++ (0,-0.25) -- (#1,{-#2+0.25});
}

\newcommand\singlet[3]{
  \node at (#1,-#2+1) {\(#3\vphantom{y^2}\)};
}

\newcommand\framesinglet[3]{
  \node at (#1,-#2+1) {\framebox{\(#3\vphantom{y^2}\)}};
}

\newcommand\diagpair[4]{
  \node at (#1,-#2+1) {\(#3\vphantom{y^2}\)};
  \node at ({#1+1},{-#2}){\(#4\vphantom{y^2}\)};
  \draw [->] (#1,-#2+1) ++ (0.1,-0.25) -- ({#1+0.9},{-#2+0.25});
}

\newcommand\axisnum[1]{
  \node at (0.5,-#1+1) {\(#1\)};
}

\begin{center}
\begin{tikzpicture}[x=5.75em,y=2.8em]

\node at (0,-1) {\(\deg=\)};
\axisnum2
\axisnum3
\axisnum4
\axisnum5
\axisnum6
\axisnum7
\axisnum8
\node at (0.5,-8) {\(\vdots\)};
\draw (0.7, -0.5) -- (0.7, -8.5);

\singlet 12 x
\singlet 22 y

\downpair 13 {f}{x^2}
\downpair 23 g{xy}
\singlet 34 {y^2}
\diagpair 15 {fx}{x^3}
\diagpair 25 {gx+fy}{x^2y}
\framesinglet 35 {gx-fy}
\downpair 45 {gy}{xy^2}
\diagpair 16 {fg}{x^2g-xyf}
\singlet 56 {y^3}
\downpair 17 {x^2f}{x^4}
\downpair 37 {x^2g+xyf}{x^3y}
\downpair 47 {y^2f+xyg}{x^2y^2}
\framesinglet 57 {y^2f-xyg}
\downpair 67 {y^2g}{xy^3}

\singlet 28 {\cdots}
\singlet 58 {\cdots}

\singlet {3.5}9 {\vdots}
\end{tikzpicture}
\end{center}

\caption[A non-formal example of a minimal-model CDGA]{A picture of the low-degree portion of our non-formal example of a minimal CDGA. The differential is indicated by arrows, so that the long and short representations are visible. We have been sloppy about including coefficients; the indicated relationships are true up to overall scalars.}
\label{informal}
\end{figure*}
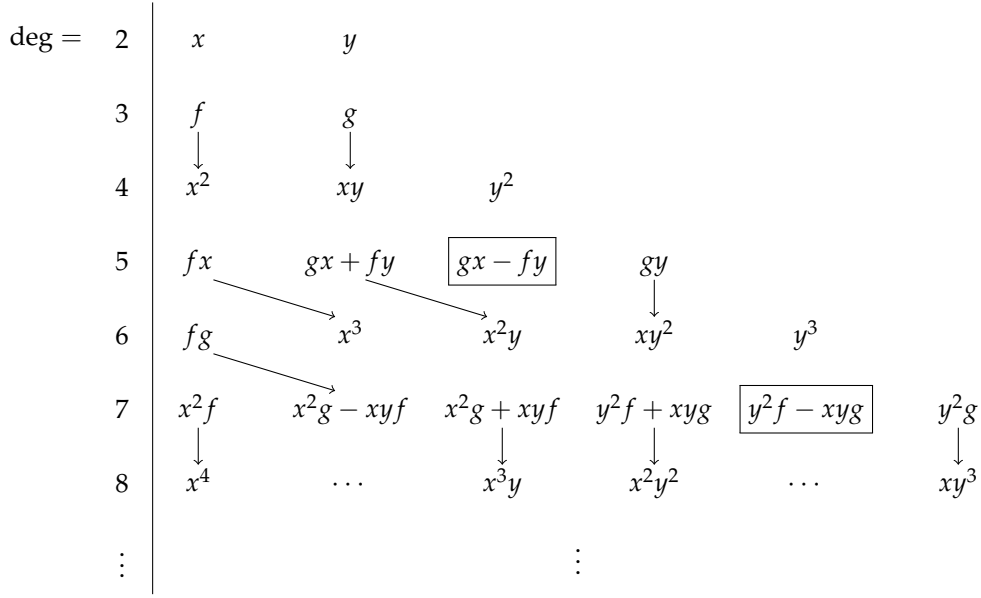

We draw a picture of this example in \protect\cref{informal}. To give an explanation which may be more tangible to the reader than the formal definition, Massey products have to do with entanglement: in a minimal model of a CDGA, they are supersymmetric vacuum states that are entangled between different oscillator degrees of freedom. 

To say this at more length, recall that entanglement is a property of a state in a Hilbert space, when viewed with respect to a basis that exhibits that Hilbert space as a tensor product of factors corresponding to subsystems. 

In a minimal-model CDGA~$M$, there are two natural choices of basis: one is the tensor-product basis corresponding to the different oscillator degrees of freedom that make up the system. With respect to this basis, $M$ is freely generated as a graded polynomial algebra. 
The other basis is the basis consisting of energy eigenstates, which shows that $M$ is a direct sum of irreducible representations of supersymmetry.

One of these bases is natural from the standpoint of the action of the de~Rham algebra, and the other is natural from the standpoint of the multiplication. These bases may not agree; if this happens, the result is that the product of operators corresponding to zero-energy states may not be an energy eigenstate at all, but rather may be a quantum superposition of zero-energy and positive-energy states. For the same reason, a zero-energy state may not be a tensor product state of different oscillators, but may only appear as an entangled state. As always, studying the appropriate picture (\protect\cref{informal}) should make this clear.

In a supersymmetric quantum mechanics whose Hilbert space is a Sullivan minimal-model CDGA, the Hilbert space is that of a finite collection of bosonic and fermionic harmonic-oscillator degrees of freedom. The cohomology classes corresponding to Massey products, like the state \(gx - fy\) in our example, are vacua of this system in which the oscillator subsystems exhibit nontrivial entanglement with one another. 
As is clear from \protect\cref{informal}, no tensor product state of the form
\[ \ket f \otimes \ket\cdots \]
can be a vacuum state, since \(\ket f\) is $Q$-exact. However, an entangled combination of states like this one \emph{is} a new supersymmetric vacuum state.

Let us also say a few words about how one should think about the meaning of homotopy equivalence from a physical perspective. Let~$f$ and~$g$ be two homomorphisms between CDGAs~$A$ and~$B$. The two are \emph{homotopic} if there exists a map $H:A \rightarrow B$ of degree $-1$, such that
\[ f - g = d_Bh + hd_A. \]
Now, an example of a homotopy equivalence between CDGAs $A$ and~$B$ is a pair of projection and inclusion maps
\[ A \xrightarrow{p} B, \quad B \xhookrightarrow{i} A, \]
such that $ip$ is homotopic to $1_A$; in other words, $ip-1_A$ (which is projection onto the orthogonal complement of $B\subset A$) is $Q$-exact.

The reader may wonder why the inclusion of the cohomology into a CDGA as the set of zero-energy states does not necessarily define a homotopy equivalence between the two. The answer is that the maps $i$ and~$p$ above are required to be homomorphisms of algebras; therefore, one cannot simply project onto any linear subspace. The kernel of~$p$ must be an ideal.

This requirement has a clear physical interpretation as well: it corresponds to our normal ideas about integrating out degrees of freedom. If the map $p$ sends a state $\ket\psi$ to zero, indicating (for instance) that $\ket\psi$ has an energy which is above a certain cutoff value, it should also send to zero states obtained by adding additional particles to~$\ket\psi$. In the context of CDGAs that are Fock spaces (that is to say, Sullivan algebras that are not necessarily minimal), one should only discard states corresponding to an entire harmonic-oscillator degree of freedom, when they all have nonzero energy (so that that degree of freedom is massive and can be integrated out).

It is therefore possible to think about (at least a subset of) homotopy equivalences between CDGAs, at least heuristically, as renormalization group flows. Since the RG scale is a continuous parameter of the physical theory, this is in line with the intuitive idea that a homotopy between physical theories should be a path in the moduli space.

Now let us review the main result of~\cite{DGMS}:
\begin{theorem}
Compact K\"ahler manifolds are formal.
\end{theorem}
\begin{proof}
With the technology we have set up, the proof is almost trivial. The de~Rham complex of a K\"ahler manifold exhibits \(\mathcal N=2\) supersymmetry, either with respect to supercharges \(\partial\) and~\(\bar\partial\) or with respect to \(\dif{}\) and \(\dif{}\conj \doteq  i \left( \partial - \bar\partial \right)\). The latter choice corresponds to 
\( t = [1:1] \)
in our \(\CP1\) family. From the standpoint of complex coefficients (and therefore physics) there is no difference; however, \cite{DGMS} choose \(\dif{}\) and~\(\dif{}\conj\) because they are real operators, and so act on the real de~Rham complex. 

There is no difference between \(\dif{}\)-cohomology and \(\dif{}\conj\)-cohomology; both count short representations of \(\mathcal N=2\) supersymmetry. However, passing to \(\dif{}\conj\)-cohomology furnishes an \emph{equivalence} of CDGAs with respect to the differential~\(\dif{}\):
\begin{equation}
(\Omega, \dif{}) \leftarrow (\ker \dif{}\conj, \dif{}) \to (H_{\dif{}\conj}, \dif{}).
\end{equation}
It should be clear that both arrows are quasi-isomorphisms, and that the differential induced on \(\dif{}\conj\)-cohomology by~\(\dif{}\) is identically zero. The result follows immediately.
\end{proof}

Formality thus follows simply whenever \emph{more than one} supercharge can be used to identify the zero-energy spectrum of the \emph{same} Hamiltonian. This is precisely the context of $\mathcal N\ge 2$ supersymmetry. However, the reader should keep in mind that the absence of enhanced supersymmetry does not necessarily imply a lack of formality.

One of the initial motivations that led to this study was as follows: Not every compact manifold is formal. A large class of examples are nilmanifolds, which are obtained as quotients of nilpotent simply connected Lie groups by discrete co-compact subgroups. 
For instance, we can take the three-dimensional Heisenberg group consisting of matrices 
\begin{equation}
N = \Set{ \begin{pmatrix} 1 & x & z \\ 0 & 1 & y \\ 0 & 0 & 1 \end{pmatrix}\colon x,y,z \in \mathbb R },
\end{equation}
which is homeomorphic to~\(\mathbb R^3\), and the subgroup consisting of elements for which \(x,y,z\in\mathbb Z\). The quotient is a compact three-manifold that is not formal.
One can construct simply connected manifolds with the same properties.

Since the ``vanilla'' flavor of supersymmetric quantum mechanics
does not require any additional structure, K\"ahler or otherwise, we were led to ask: does physics look different on a non-formal manifold? Do the Massey products, interpreted as operations relating sets of vacuum states, have a physical significance or meaning? In light of the discussion above, the reader should immediately see the proper response to this question. Massey products arise because of entanglement between the basis of energy eigenstates (or supersymmetry representations) and the basis that is most natural with respect to the wedge product. However, as we remarked in the introduction, the wedge product is not an obviously meaningful operation in supersymmetric quantum mechanics, and so the question does not make sense without additional structure. 

It is then natural to look for other examples of physical systems for which certain sets of \emph{operators} admit a description in terms of the cohomology of a supercharge, and therefore higher-order Massey products between (for instance) $Q$-cohomology classes of local operators, or supersymmetric defects with respect to the product defined by fusion~\cite{Brunner-Roggenkamp-Rossi}, might occur.
The results of~\cite{DGMS} translate directly to a particular class of physical examples: \(\mathcal N=2\) supersymmetric sigma models in two dimensions, which require K\"ahler structure on the target space as a necessary and sufficient condition for supersymmetry. These sigma models admit two different well-known topological twists (the $A$- and $B$-twists), after which the TQFT Hilbert spaces are identified with the cohomology of a geometric CDGA associated to the target space. The statement that this CDGA must necessarily be formal (for the $B$-model, this generalization of~\cite{DGMS} was proven by Zhou~\cite{Zhou}) suggests that whatever physical phenomena higher cohomology operations may represent cannot occur between local operators in a two-dimensional supersymmetric sigma model---at least for the $A$-twist in the large volume limit where the product agrees with the wedge product. The situation is more subtle for the $A$-model at finite volume. 

Understanding the correct notion of formality for quantum field theories is more subtle than for CDGAs or $0+1$-dimensional systems. The reason is that, while spaces of BPS operators, states, or defects form algebras on very general grounds~\cite{Harvey-Moore,LVW}, and moreover their algebraic structure comes from natural and meaningful physical operations (such as the structure of Fock space in perturbative theories, the OPE, or defect fusion), these operations are only guaranteed to be nonsingular for BPS objects, and are highly singular in general. Formality is not a property of the cohomology ring, but rather a property of the differential algebra from which it is derived; as such, understanding the physics (and the algebraic structure) of the non-BPS part of the theory is critical. Nonetheless, in quantum field theories that admit topological twists, one can look at the portion of the supersymmetry algebra that includes the scalar supercharge, and sometimes prove analogues of the \(\dif\dif{}\conj\) lemma. This can be done whenever the \(\mathcal N=2\) algebra of quantum mechanics is a subalgebra of the higher-dimensional twisted supersymmetry algebra. We will return to this question in Section~\ref{twists}.

It is also interesting to notice that the above discussion suggests a way to define ``homotopy groups'' with complex coefficients, for \(\mathcal N=2\) supersymmetric quantum mechanics, or more generally for twists of theories for which a $\dif\dif{}\conj$-lemma can be proven. These homotopy groups are really complex vector spaces, just like the $\mathbb C$-de~Rham homotopy groups we discussed above. To calculate them, one should take the graded algebra of $Q$-cohomology classes or BPS states, view it as a CDGA with zero differential, and construct its minimal model as was described above. The graded components of the minimal-model CDGA are then the $\mathbb C$-de~Rham homotopy groups; when the theory one started with is the \(\mathcal N=2\) quantum mechanics of a particle on a K\"ahler manifold, this procedure recovers the homotopy groups of the target space, tensored with~$\mathbb C$. For more abstract examples of supersymmetric quantum mechanics, the meaning of these new invariants is less clear.


\section{Physical origins of Lefschetz operators}
\label{fofks}
In this section, we quickly review work of Figueroa-O'Farrill, K\"ohl, and Spence~\cite{FOFKS}, who showed (pursuing an idea due to Witten) that the action of Lefschetz operators in the de~Rham complex of K\"ahler and hyper-K\"ahler manifolds can be understood in terms of dimensional reduction of the higher-dimensional supersymmetric sigma models that can be defined on these spaces. Our goal in doing this is to emphasize the utility of studying these theories via dimensional reduction to supersymmetric quantum mechanics, and also to point out that their argument allows analogues of the Lefschetz action to be defined in theories that are not necessarily sigma models. These Lefschetz-type symmetries are precisely the $R$-symmetries of the dimensionally reduced theories.
We detail the consequences of these Lefschetz-type symmetries, comparing them with the consequences of two-dimensional superconformal $R$-symmetries.\footnote{This is motivated by the fact that two-dimensional sigma models on Calabi-Yau and hyper-K\"ahler manifolds are automatically superconformal.} 
Further, we comment on some intriguing and (we believe) unexplained numerical coincidences.

The central idea is simple to state. Suppose that a \(d\)-dimensional, minimally supersymmetric sigma model can be defined on a target manifold~\(X\). When \(d=4\), it is necessary and sufficient for \(X\) to be K\"ahler; when \(d=6\), \(X\) must be hyper-K\"ahler. Again, this can be seen just by counting the supercharges that act in the de~Rham complex.

Ignoring questions of signature, the Lorentz group of the theory will be~$\SO(d)$. In order to dimensionally reduce to a quantum mechanics problem, we must fix a splitting of the worldsheet coordinates into one time and $d-1$ spatial directions. This breaks the Lorentz group to $\SO(d-1)$, acting on the spatial directions; this symmetry should survive as a flavor symmetry in the dimensionally reduced theory. 

For the four-dimensional sigma model, we therefore expect the symmetry algebra $\so(3) \cong \su(2)$ to act; for the six-dimensional sigma model, the relevant algebra will be~$\so(5)$. The result of~\cite{FOFKS} is that, in the context of K\"ahler and hyper-K\"ahler sigma models, these group actions agree precisely with the Lefschetz actions we reviewed above.

However, when arguing that these symmetry algebras must act on quantum mechanics after dimensional reduction, we did not actually use the fact that the theories in question are sigma models. The argument is quite general, and applies to the supersymmetric quantum mechanics obtained after dimensional reduction of \emph{any} theory. As such, if we are interested in seeing how close such a supersymmetric quantum mechanics problem is to the algebraic model of an honest space, we can use the existence of a Lefschetz action without making any further assumptions. 

\begin{figure*}[t]
\centerline{
\hspace{-5em}
\xymatrix{
&\makebox[\width][c]{\begin{tabular}{@{}c@{}}
\(\SO(5)\) \\
hyper-K\"ahler; 6d red.
\end{tabular}} 
\ar[dr] \ar[dl] & \\
\makebox[\width][c]{\begin{tabular}{@{}c@{}}
$\SO(4) \cong \SU(2)\times\SU(2) $ \\ \(\mathcal N=(4,4)\); 5d red. 
\end{tabular}}
\ar[d] \ar[drr] & & 
\makebox[\width][c]{\begin{tabular}{@{}c@{}}
$\SU(2) \times \operatorname U(1)$ \\ K\"ahler
\end{tabular}} \ar[dll] \ar[d] \\
\makebox[\width][c]{\begin{tabular}{@{}c@{}}
$\SO(3) \cong \SU(2) $\\4d red.
\end{tabular}}
\ar[dr] & & 
\makebox[\width][c]{\begin{tabular}{@{}c@{}}
$\operatorname U(1) \times\operatorname U(1)$ \\ \(\mathcal N=(2,2)\)
\end{tabular}} \ar[dl] \\
&
\makebox[\width][c]{\begin{tabular}{@{}c@{}}
$\SO(2) =\operatorname U(1)$ \\ Riemannian; 3d red.
\end{tabular}} & \\
}}
\caption[A subgroup diagram for $R$-symmetries of various theories]{The hierarchy of group actions that occur on twists of 2d sigma models, on dimensional reductions of higher-dimensional field theories, and on the cohomology of target spaces.}
\label{hierarchy}
\end{figure*}
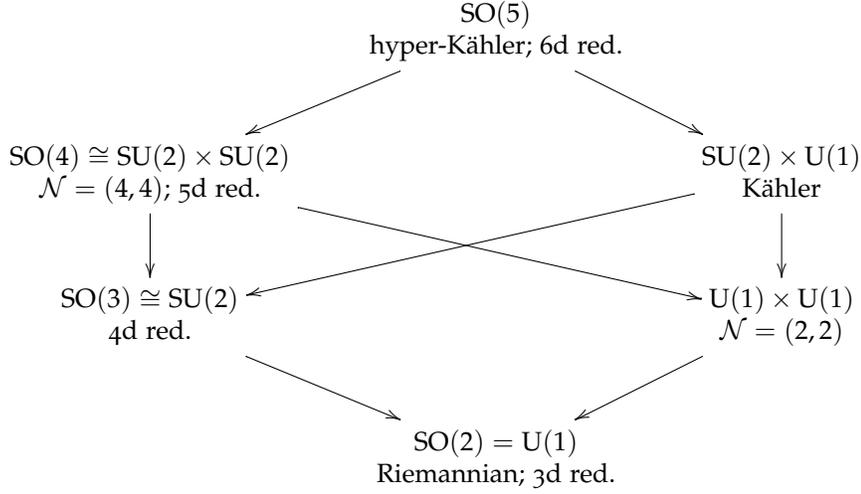

To be concrete, consider the \(\mathcal N=1\) supersymmetry algebra in four dimensions:
\begin{equation}
\{ Q_\alpha, \bar Q_{\dot\beta}\} =2 i\sigma^\mu_{\alpha\dot\beta}P_\mu
\label{4dN=1}
\end{equation}

Following \cite[Table 1]{FOFKS}, make the following identification:
\begin{align*}
(\partial,\bar\partial^\dagger)&\mapsto(Q_1,Q_2)&
(\partial^\dagger,\bar\partial)&\mapsto(\bar Q_{\dot1},\bar Q_{\dot2}).
\end{align*}
Note that there is a \(\CP1\) family of choices to be made here. We are picking a point out of the projectivization of the Weyl spinor space \(\mathbb C^2\). This is the same as picking a direction out of three-dimensional Euclidean space.

It is then simple to rewrite the algebra~\protect\eqref{4dN=1} as follows:
\begin{gather}
(\Delta_\partial,\Delta_{\bar\partial})=(2(P_0+P_3),2(P_0-P_3)), \nonumber \\
\{\partial,\bar\partial\}=\{Q_1,\bar Q_{\dot2}\}=P_1+ iP_2, \\
\{ \partial, \bar\partial^\dagger \} = \{ Q_1, Q_2 \} = 0. \nonumber
\label{4dInformal}
\end{gather}
Upon dimensional reduction to $0+1$ dimensions, we set the momentum operators \(P_i=0\). After doing so, this becomes the algebra of \(\mathcal N=2\) quantum mechanics, which we have written in the notation appropriate to K\"ahler manifolds.

If we do \emph{not} dimensionally reduce to quantum mechanics, this algebra is not an \(\mathcal N=2\) algebra. However, using Lorentz invariance, we may without loss of generality take the momentum to lie in the $P_3$ direction. We then obtain the \(\mathcal N=1^2\) algebra, where the Laplacians are as in~\protect\eqref{4dInformal}.

Recall from our discussion of the \(\mathcal N=1^2\) algebra above that the representations that spoil the \(\dif\dif{}\conj\)-lemma are the two-dimensional representations that are short with respect to one supercharge and long with respect to the other. By inspection of~\protect\eqref{4dInformal}, these are states for which $P_0 = \pm P_3$: that is to say, massless states in the four-dimensional theory.

Returning to quantum mechanics, the point is that we can naturally identify an action of $\su(2)\cong\so(3) $ that is compatible with the \(\mathcal N=2\) supersymmetry algebra we identified above. We define the analogues of the Lefschetz operators as follows:
\begin{align}
D &=2J_3 & (\Lambda^\dagger,\Lambda)=J_\pm&=J_1\pm iJ_2.
\label{abstractLefschetz}
\end{align}
These operators commute with the Hamiltonian  $P_0$ by the standard Poincar\'e algebra:
\[
[M_{\mu\nu}, P_\rho] = i \left( \eta_{\mu\rho} P_{\nu} - \eta_{\nu\rho} P_{\mu} \right) . 
\]
Again, although~\protect\eqref{abstractLefschetz} agrees with the usual Lefschetz action in the context of sigma models, nothing in this derivation relies on having a sigma-model description of the field theory! The analogues of the crucial K\"ahler identities~\protect\eqref{Kahler1} follow from the fact that the supercharges transform as a spinor under rotations.


Given the above, it should be clear that any supersymmetric quantum mechanics arising from dimensional reduction from a field theory in~$d$ dimensions has a Lefschetz-type action of~$\SO(d-1)$. In Figure~\ref{hierarchy}, we make a subgroup diagram showing these actions together with those that follow from $R$-symmetry in two-dimensional superconformal theories and those that follow from geometric realization as a sigma model. Each subgroup corresponds to a specific element on the partially ordered set of geometrical consequences.

This partially ordered set is not totally ordered. The consequences of \(\mathcal N=(2,2)\) superconformal symmetry and realizability as a K\"ahler sigma model are incomparable, for instance. Since $\SU(2)$ acts vertically on the Hodge diagram for K\"ahler manifolds, but the $\SU(2)\times\SU(2)$ arising from $\mathcal N=(4,4)$ $R$-symmetry acts diagonally --- one is not a subgroup of the other.

Furthermore, notice that the conjunction of these two requirements is strictly weaker than the consequences of realizability as a hyper-K\"ahler sigma model. For instance, \protect\cref{octagon-no-degeneracy} is consistent with \(\mathcal N=(2,2)\) superconformal symmetry and the K\"ahler \(\SU(2)\) Lefschetz action, but is inconsistent with \(\SO(5)\) Lefschetz action  (since the diagram has the wrong degeneracies in the interior of the octagon).

\begin{figure*}[t]
\centering
\begin{subfigure}{0.39\textwidth}
\centering
\begin{tikzpicture}[scale=0.85]
\drawdiamond3

\node at (1,0) {\(\bullet\)};
\node at (-1,0) {\(\bullet\)};
\node at (0,-1) {\(\bullet\)};
\node at (0,1) {\(\bullet\)};
\node at (1,-2) {\(\bullet\)};
\node at (1,2) {\(\bullet\)};
\node at (-1,-2) {\(\bullet\)};
\node at (-1,2) {\(\bullet\)};
\node at (-2,1) {\(\bullet\)};
\node at (2,1) {\(\bullet\)};
\node at (-2,-1) {\(\bullet\)};
\node at (2,-1) {\(\bullet\)};
\end{tikzpicture}
\caption{This pattern is consistent with \(\mathcal N=(4,4)\) superconformal algebra's \(\SU(2)\times\SU(2)\) action, but is incompatible with
4d Lefschetz action.}
\end{subfigure}
\qquad
\begin{subfigure}{0.39\textwidth}
\centering
\begin{tikzpicture}[scale=1.2]
\drawdiamond2
\node at (0.5,0.5) {\(\bullet\)}; 
\node at (-0.5,0.5) {\(\bullet\)};
\node at (0.5,-0.5) {\(\bullet\)};
\node at (-0.5,-0.5) {\(\bullet\)};
\node at (1.5,0.5) {\(\bullet\)};
\node at (1.5,-0.5) {\(\bullet\)};
\node at (-1.5,0.5) {\(\bullet\)};
\node at (-1.5,-0.5) {\(\bullet\)};
\node at (0.5,1.5) {\(\bullet\)};
\node at (0.5,-1.5) {\(\bullet\)};
\node at (-0.5,1.5) {\(\bullet\)};

\node at (-0.5,-1.5) {\(\bullet\)};
\node at (0,-2) {\(\phantom{\bullet}\)};
\node at (0,2) {\(\phantom{\bullet}\)};
\end{tikzpicture}
\caption{This octagon is compatible with \(\mathcal N=4\) with 4d Lefschetz action, but is incompatible with 6d Lefschetz action.}\label{octagon-no-degeneracy}
\end{subfigure} 
\caption{Illustration of the differences between requirements tabulated in \protect\cref{hierarchy}.}
\end{figure*}
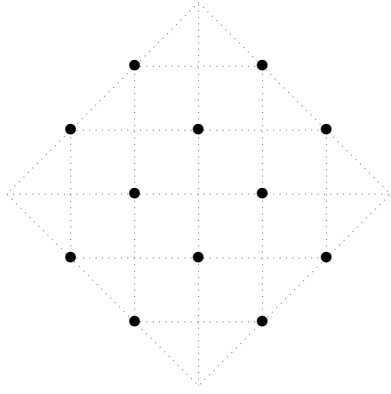
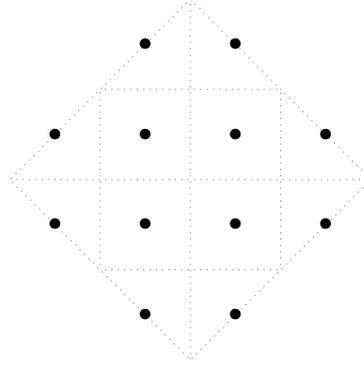


There is one interesting thing to note. Verbitsky constructed the $\SO(4,1)$ action for hyper-K\"ahler manifolds by using the three K\"ahler classes as generators. Since these three classes are independent, $b_2\ge 3$ for any hyper-K\"ahler manifold. However, if $b_2$ is large, there are many more ``candidates'' for defining Lefschetz operators, and one may wonder if a larger symmetry algebra can be defined in this case. 

In later work, Verbitsky~\cite{Verbitsky-enhanced, Verbitsky-enhanced-proof} showed that this is indeed the case. He constructed an action of~$\SO(4, b_2-2)$ on the cohomology of a hyper-K\"ahler manifold, and it is natural to wonder (indeed, the authors of~\cite{FOFKS} already do so) whether this can be understood from a supersymmetric perspective as well.

The only compact hyper-K\"ahler four-manifold is the K3 surface. For hyper-K\"ahler manifolds of (real) dimension eight, Guan~\cite{Guan} has shown that the second Betti number must be one of the following:
\[b_2\in\{3,4,5,6,7,8,23\},\]
although examples are only explicitly known for \(b_2=7,23\). We draw the Hodge diamonds of these examples~\cite{Salamon} in Figure~\ref{hypK-hodge}.

If one were to realize Verbitsky's \(\so(b_2+2)\) actions for these target spaces in terms of Lorentz symmetry, as has been done for the standard Lefschetz actions, this would correspond to the following set of dimensions for the corresponding sigma models:
\[d\in\{6,7,8,9,10,11,26\}.\]
The largest three---$10$, $11$, and $26$---happen to equal the spacetime dimensions of superstring theory, M-theory, and bosonic string theory. Moreover, the examples that have been constructed correspond exactly to $d=10$ and~$d=26$. We emphasize, however, that these should be thought of as worldsheet (not target-space) dimensions!

It is worth pointing out that compactification of superstring theory on a hyper-K\"ahler eight-manifold leaves a two-dimensional effective theory in the target space. In some sense, it is as if a rather bizarre duality exchanges the roles of target space and worldsheet. 

To the best of our knowledge, there is no straightforward explanation of these facts, since there is no such thing as the supersymmetric sigma model in dimension $d>6$. (This can be seen in many ways: from Berger's classification of exceptional holonomy, or by noting that the algebra of complex structures required for~\(\mathcal N=8\) would have to be that of the octonions.) In such a large spacetime dimension, the spinor representations in supersymmetry would be far bigger than the number of differential operators present on differential forms on a hyper-K\"ahler manifold. 

These considerations suggest a scenario in which a \(d\)-dimensional theory, which is not supersymmetric, somehow acquires supersymmetry upon dimensional reduction to six (or fewer) dimensions. However, we lack any concrete proposal or meaningful evidence for this speculative scenario.

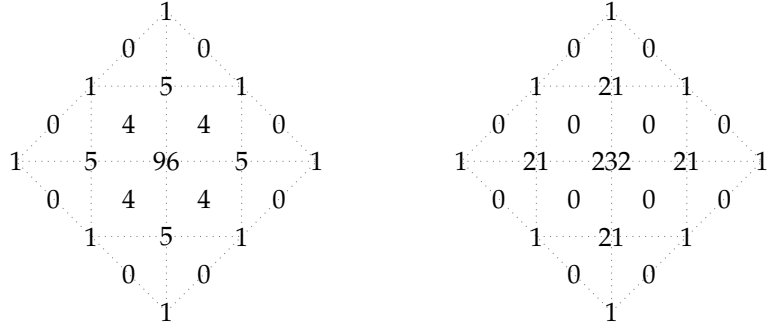
\begin{figure*}[t!]
\centering
\newcommand\hksymmetricput[3]{
  \node at (#1,#2) {\(#3\)};
  \node at (#1,-#2) {\(#3\)};
  \node at (-#1,#2) {\(#3\)};
  \node at (-#1,-#2) {\(#3\)};
  \node at (#2,#1) {\(#3\)};
  \node at (#2,-#1) {\(#3\)};
  \node at (-#2,#1) {\(#3\)};
  \node at (-#2,-#1) {\(#3\)};
}

\newcommand\hksymmetricputDiag[2]{
  \node at (#1,#1) {\(#2\)};
  \node at (#1,-#1) {\(#2\)};
  \node at (-#1,#1) {\(#2\)};
  \node at (-#1,-#1) {\(#2\)};
}

\newcommand\hksymmetricputHV[2]{
  \node at (0,#1) {\(#2\)};
  \node at (0,-#1) {\(#2\)};
  \node at (#1,0) {\(#2\)};
  \node at (-#1,0) {\(#2\)};
}

\newcommand\drawbordernumbers{
  \drawdiamond2
  \hksymmetricputHV{2}{1}
  \hksymmetricputDiag{1}{1}
  \hksymmetricput{0.5}{1.5}{0}
}

\begin{tikzpicture}
\drawbordernumbers
\hksymmetricputHV15
\hksymmetricputDiag{0.5}4
\node at (0,0) {\(96\)};
\end{tikzpicture}
\qquad\qquad
\begin{tikzpicture}
\drawbordernumbers
\hksymmetricputHV1{21}
\hksymmetricputDiag{0.5}0
\node at (0,0) {\(232\)};
\end{tikzpicture}

\caption[Hodge diamonds of compact hyper-K\"ahler eight-manifolds]{The Hodge diamonds of known examples of compact hyper-K\"ahler eight-manifolds with $b_2 > 3$.} 
\label{hypK-hodge}
\end{figure*}

\section{Topological twists of quantum field theories}
\label{twists}

In this section, we take inspiration from the discussion of the $\dif\dif{}\conj$-lemma above and look at topological twists of field theories in higher dimensions. For some twists, the supersymmetry algebra after twisting contains the algebra of \(\mathcal N=2\) quantum mechanics as a subalgebra; for others, this is not true. 


\subsection{4d \(\mathcal N=2\)}

The relevant part of the four-dimensional extended supersymmetry algebra is
{\allowdisplaybreaks
\begin{gather*}
\{ Q_\alpha^A, \bar Q_{\dot \beta B} \} = 2 \sigma^\mu_{\alpha\dot\beta} P_\mu \delta^A_B, \\
\{ Q_\alpha^A, Q_\beta^B \} = \epsilon_{\alpha\beta} \epsilon^{AB} Z,
\end{gather*}}
The adjoint relationship is, once again, $(Q_\alpha^A)^\dagger = \bar Q_{\dot \alpha A}$. 
We will only consider the case without central charges, so that the scalar supercharges obtained after twisting are honestly nilpotent. In the case when~$Z$ is nonzero, the scalar supercharge will square to~$Z$; topological twists can still be understood in this setting, but we will not address this here. The reader is referred to~\cite{Labastida-Marino}.

The representations in which the supercharges sit are
\begin{align*}
Q_\alpha^A&\colon (\rep 2, \rep1; \rep 2) &
\bar Q_{\dot\beta}^B&\colon (\rep1, \rep 2; \rep 2).
\end{align*}
The first two numbers refer to the representation of the Lorentz group $\SO(4) \cong \SUl \times \SUr$, and the last number is a representation of $\SU(2)$ $R$-symmetry.
Up to equivalence, there is a unique way to construct a topological twist of this algebra. One chooses the homomorphism 
\[ \pr_1\colon \SUl \times\SUr \to \SU(2)_R, \]
and lets the Lorentz group act on $R$-symmetry indices via this identification rather than trivially. The resulting representations are
{\allowdisplaybreaks
\begin{align*}
\left(Q_1^1, \frac{1}{\sqrt2}(Q_1^2 + Q_2^1), Q_2^2 \right) &\colon ({\rep 3}, 1) \\
q \doteq \frac{1}{\sqrt{2}} (Q_1^2 - Q_2^1) &\colon  (1,1)  \\
\bar Q_{\dot\beta}^B &\colon (\rep 2,{\rep 2}).
\end{align*}
}
The quantum numbers indicated are for the new Lorentz group $\SUl' \times \SUr$. The scalar supercharge with respect to which we twist is~$q$; in order that $q$ be nilpotent, we must set the central charge to zero. It is simple to check that the ``Laplace operator'' corresponding to this supercharge is
\begin{align*}
\{ q, q^\dagger \} &= \frac{1}{2} \{ Q_1^2 - Q_2^1, \bar Q_{\dot 1 2} - \bar Q_{\dot 2 1} \} \\
&= \frac12 \left( 2 P_\mu \sigma^\mu_{1 \dot 1} + 2 P_\mu \sigma^\mu_{2 \dot 2} \right) \\
&= \tr (P_\mu \sigma^\mu_{\alpha\beta}) \\
&= 2 P_0.
\end{align*}
Therefore, the $q$-cohomology counts states on which the Hamiltonian acts by zero. This is what we should expect if the result of twisting is to be a topological theory. 

There is no other scalar supercharge in the theory, so adding any more supercharges to our subalgebra will give a result that is not Lorentz invariant. Of course, this is not a problem: one merely needs the existence of an \(\mathcal N=2\) subalgebra to show that the $\dif\dif{}\conj$-lemma holds, whereas a single Lorentz-invariant supercharge is enough to make the twist. It is instructive to consider the commutation relations with the supercharge
\[ \tilde q = \frac1{\sqrt2} \left( Q_1^2 + Q_2^1 \right). \]
The two supercharges commute with one another (even in the presence of central charges): 
\begin{align*}
\{ q, \tilde q \} &= \frac{1}{2} \{ Q_1^2 - Q_2^1, Q_1^2 + Q_2^1 \}  \\
&= \frac12 \left( \epsilon_{12}Z^{21} - \epsilon_{21}Z^{12} \right) \\
&= 0. 
\end{align*}
We have therefore identified a \(\mathcal N=1.5\) subalgebra. However, this fails to close into an \(\mathcal N=1^2\) algebra. It is easy to check that
\[
\{ q, \tilde q^\dagger \} = 2 P_3. 
\]
However, the Laplacian of $\tilde q$ nonetheless agrees with that of~$q$:
\[
\{ \tilde q, \tilde q^\dagger \} = 2 P_0.
\]
The reader may see a puzzle here---the Laplacians agree, so there is no way for a representation to have both zero and nonzero energy. It should follow that the one-dimensional representations are the only short representations. 

The resolution to this puzzle consists of remembering that (in Lorentzian signature) $P_0 \geq |P_3|$. Therefore, if a state satisfies $P_0=0$, the operator $P_3$, which is the staircase operator in this context, must act on it by zero as well. This is sufficient to show that the only permissible staircases have length two, just as is the case for the $\mathcal N=1^2$ algebra. 

Indeed, the algebra we have written is equivalent to the $\mathcal N=1^2$ algebra after a change of basis. Taking the supercharges to be
\begin{equation*}
\begin{split}
\partial=Q_1^2, \quad
\bar\partial=-Q_2^1, \quad  \\
\partial^\dagger=\bar Q_{\dot12}, \quad 
\bar\partial^\dagger=-\bar Q_{\dot21},
\end{split}
\end{equation*}
it follows immediately that
\begin{align*}
\Delta_\partial=\{\partial,\partial^\dagger\}&=2\sigma_{11}^\mu P_\mu=P_0+P_3\\
\Delta_{\bar\partial}=\{\bar\partial,\bar\partial^\dagger\}&=2\sigma_{2\dot2}^\mu P_\mu=P_0-P_3\\
\{\partial,\bar\partial\}&=\{\partial,\bar\partial^\dagger\}=0.
\end{align*}
This is therefore a typical \(\mathcal N=1^2\) system. The new states that appear in $\partial$- and $\bar\partial$-cohomology, as compared to $q=\partial+\bar\partial$-cohomology, satisfy \(P_3=\pm P_0\) respectively. Put differently, \emph{massless} states cause the failure of the $\dif\dif{}\conj$-lemma. 
In a topological twist of a massive four-dimensional $\mathcal N=2$ theory, we would expect it to hold.

\subsection{3d \(\mathcal N=4\)}

Inspecting the quantum-mechanical subalgebra we were considering above, it is easy to see that it becomes \(\mathcal N=2\) subalgebra upon setting $P_3=0$; in other words, after dimensional reduction to three dimensions. It makes sense that this should generate a new scalar supercharge: after the twist, the $\bar Q_{\dot\beta}^B$ are in a vector representation of~$\SO(4)'$, which will become a vector and a scalar upon dimensional reduction. The new scalar is precisely~$q^\dagger$; therefore, we obtain a complete Lorentz-invariant copy of the de~Rham algebra, sitting inside a non-Lorentz-invariant \(\mathcal N=2\) algebra. 

This shows that the $\dif\dif{}\conj$-lemma holds for one of the two possible twists of three-dimensional $\mathcal N=4$ theories. While there are two twists that are in general inequivalent~\cite{Blau-Thompson}, they are the same at the level of the supersymmetry algebra. Only the decomposition of supermultiplets differs. As such, the \(\mathcal N=2\) subalgebra we identified above exists in both possible twists.

\subsection{4d \(\mathcal N=4\)}
This is the dimensional reduction of minimal supersymmetry in ten dimensions. The $R$-symmetry is $\Spin(6)\cong\SU(4)$. The representations of the supercharges before twisting are
\begin{align}
Q_\alpha^A &\colon (\rep 2, \rep1; \rep 4) &\bar Q_{\dot\beta}^B&\colon (\rep1, \rep 2; \rep{\bar 4}).
\end{align}

\paragraph{Donaldson--Witten twist} This corresponds to choosing an \(\mathcal N=2\) subalgebra and twisting as above. The twisting homomorphism embeds $\SUl$ in the obvious way as a block diagonal inside $\SU(4)$. This leaves an~$\SU(2)\times \operatorname U(1)$ subgroup of the $R$-symmetry unbroken. The $R$-symmetry representations of the supercharges decompose as 
\begin{align}
(\rep 2,\rep1;\rep 4) &\to (\rep 3, \rep1; \rep1)_1 \oplus (\rep1,\rep1;\rep1)_1 \oplus (\rep 2,1;\rep 2)_{-1} \\
(\rep1,\rep 2;\rep{\bar 4}) &\to (\rep 2,\rep 2;\rep1)_{-1} \oplus (\rep1,\rep 2;\rep 2)_{1}
\end{align}
The representations indicated at right are for $\SUl' \times \SUr \times \left( \SU(2)_R\times\operatorname U(1)_R\right)$.
Taking $\SUl'$ to act on the first two $R$-symmetry indices, the scalar supercharge is
\[ q = \frac{1}{\sqrt 2} \left( Q_1^2 - Q_2^1 \right), \]
just as for the twist of the \(\mathcal N=2\) theory. All the calculations from above continue to be valid, so that the corresponding Laplacian is again proportional to~$P_0$.

However, it is now easy to find another choice of supercharge that defines an~\(\mathcal N=2\) subalgebra. For instance, one can take
\[ \tilde q = \frac{1}{\sqrt 2} \left( Q_1^4 - Q_2^3 \right) . \]
It is simple to check the commutation relations between $q$, $\tilde q$, and their adjoints. Once again, this subalgebra is not Lorentz-invariant, but this is not important in the context of proving the $\dif\dif{}\conj$-lemma. 

\paragraph{Vafa--Witten twist}
This corresponds to embedding \(\SUl\) diagonally inside the obvious $\SO(3)\times\SO(3)$ subgroup of~\(\SO(6)\cong\SU(4)\). The unbroken $R$-symmetry is \(\SU(2)_R\), rotating the two diagonal blocks into one another. (If we write the \(\SU(4)\) fundamental index as a pair of indices valued in~$\{1,2\}$, \(\SUl\) acts on one of these indices and the unbroken \(\SU(2)_R\) on the other.) The representations of the supercharges decompose as follows:
\begin{align}
(\rep2,\rep1; \rep4) &\to (\rep3,\rep1;\rep 2) \oplus (\rep1,\rep1;\rep 2), \\
(\rep1,\rep2; \rep{\bar 4}) &\to (\rep2, \rep 2;\rep 2).
\end{align}
The representations indicated at right are for $\SUl' \times \SUr \times \SU(2)_R$. The scalar supercharges can be written explicitly as 
\begin{align}
q_\uparrow &= \frac1{\sqrt2} \left( Q_1^2 - Q_2^1 \right) &
q_\downarrow &= \frac1{\sqrt2} \left( Q_1^4 - Q_2^3 \right).
\end{align} 
(The subscripts refer to the representation of~$\SU(2)_R$.) 

These supercharges are the same as the $q$ and~$\tilde q$ identified above, and so satisfy the same commutation relations, generating an~\(\mathcal N=2\) subalgebra.

\paragraph{Marcus--Kapustin--Witten twist}
This corresponds to the obvious block-diagonal homomorphism 
\[ \SUl \times \SUr \to \SU(4)_R. \]
There is an unbroken $\operatorname U(1)_R$ symmetry, commuting with this embedding, which rotates the overall phases of the two blocks in opposite directions. 
Under this embedding the $\rep 4$ of the $R$-symmetry group transforms as $(\rep 2, 1)_1 \oplus (1,\rep 2)_{-1}$ with respect to the twisted Lorentz group and the unbroken $\operatorname U(1)_R$.  The transformation of the $\rep{\bar 4}$ representation is the same, but with $\operatorname U(1)$ charges reversed.

It follows that the supersymmetries transform after the twist as 
\begin{equation}
\begin{aligned}
(\rep 2, \rep1; \rep 4) &\to (\rep1,\rep1)_1 \oplus (\rep 3,\rep1)_1 \oplus (\rep 2,\rep 2)_{-1} \\
(\rep1, \rep 2; \rep{\bar 4}) &\to (\rep1,\rep1)_1 \oplus (\rep1,\rep 3)_1 \oplus (\rep 2, \rep 2)_{-1}.
\end{aligned}
\end{equation}
The representations indicated at right are for $\SUl' \times \SUr' \times \operatorname U(1)_R$. 
(Our convention for the sign of the $\operatorname U(1)_R$-charge is the opposite of that of~\cite{Kapustin-Witten}, so that the scalar supercharges carry charge $+1$. This is in keeping with our use of cohomological grading conventions, for which the differential has positive degree, throughout the paper.)

Explicitly, the scalar supercharges can be written as 
\begin{equation}
q\Left = \frac{1}{\sqrt{2}} \left( Q_1^2 - Q_2^1 \right), \quad
q\Right = \frac{1}{\sqrt{2}} \left( \bar Q_{\dot 1 4} - \bar Q_{\dot 2 3} \right).
\end{equation}
The subscripts indicate the chirality of the spinor supercharge from which the scalar is derived. It is clear that these supercharges and their adjoints form the same subalgebra as for the Vafa-Witten twist, since $q\Left = q_\uparrow$ and $q\Right = q_\downarrow^\dagger$. 

In fact, the VW and MKW twists are isomorphic for the sphere \(S^3\) Hilbert space (equivalently, for local operators). These are just \(\tr\phi^n\) in both cases, where \(\phi\) is the adjoint sgaugino. The two twists are, however, very different for non-spheres, where the TQFT is \emph{not} a subsector of the untwisted SQFT.

\subsection{2d superconformal algebras}

As we already mentioned above, the $\mathcal N=(2,2)$ superconformal algebra has a global subalgebra (for instance, in the Ramond sector) which is precisely equal to two copies of the de~Rham algebra, one in the left-moving and one in the right-moving sector. As such, the algebra relevant to twists of these theories is always the $\mathcal N=1^2$ algebra.

In this case, one can consider deformation problems of the sort we discussed in general, where some mixture of these four supercharges defines a nilpotent operator. This leads to the $A$- and $B$-twists, as well as to the elliptic genera of $\mathcal N=(2,2)$ theories at the special points where BPS degeneracies are enhanced. Further remarks on these ideas in the context of $\mathcal N=(2,2)$ theories will be made in~\cite{GNSSS-Forthcoming}.

All twists of $\mathcal N=(4,4)$ theories come from $\mathcal N=(2,2)$ subalgebras: there is a \(\CP1\) family of twists~\cite{Berkovits-Vafa}, and this corresponds to choosing a U(1) inside SU(2) R-symmetry, which amounts to choosing \(\mathcal N=2\) subalgebra of \(\mathcal N=4\).

\section{Superconformal symmetry and reconstructing target spaces}
In this section, we briefly note a couple of the results of Sullivan and others that motivated us to begin reading about rational homotopy theory. 

The theorems concern when an abstract rational homotopy type (presented, for instance, by a minimal-model CDGA) is, in fact, the rational homotopy type of an honest manifold. As such, one should think of them as stating sufficient conditions for an abstract algebraic model to be ``geometrical.'' This is very close to the question we raised in the introduction: given (for instance) a two-dimensional $\mathcal N=(2,2)$ superconformal theory, which one may think of as a string theory background, what conditions are sufficient to ensure that it can be described as the sigma model on some geometrical target space?

In the context of Theorem~\ref{Sullivan-Barge}, one enumerates a list of all structures that can be defined on a CDGA originating from geometry, and carefully details the properties that these structures have. This list of necessary conditions turns out to also be sufficient.

\begin{theorem}[Sullivan--Barge; see~\cite{Sullivan-Infinitesimal}, as well as~\cite{Felix}, Theorem 3.2]
\label{Sullivan-Barge}
Let \(M\) be a simply connected Sullivan minimal model over~$\mathbb Q$, whose cohomology satisfies Poincar\'e duality with respect to a top form in dimension~$d$. Choose an element $p \in \oplus_i\mathrm H^{4i}(M)$, representing the total rational Pontryagin class. 

 If \(4\nmid d\), there is a compact simply connected manifold whose rational homotopy type has~$M$ as its minimal model and~$p$ as its total Pontryagin class.

If \(d=4k\), the statement remains true if the manifold is permitted to have one singular point. The singular point may be removed if, and only if:
\begin{itemize}
\item the intersection form is equivalent over~$\mathbb Q$ to a quadratic form $\sum \pm x_i^2$; and
\item either the signature is zero, or a choice of fundamental class can be made such that the Pontryagin numbers are integers satisfying certain congruences. 
\end{itemize}
\end{theorem}
\begin{theorem}[\cite{Sullivan-Infinitesimal}, Theorem~12.5]
The diffeomorphism type of a simply connected K\"ahler manifold is determined by
its integral cohomology ring and total rational Pontryagin class, up to a finite list of possibilities.
\end{theorem}

One key feature in the above theorems is the choice of coefficients: integers or rational numbers, rather than complex numbers, which are most natural from the standpoint of physics. We were thus led to ask whether one could recover a natural integer lattice in the physical Hilbert space of a sigma model, corresponding to the image of the integral cohomology in cohomology with complex coefficients. In fact, another motivation led us to think about this as well. In considering the geometrical symmetries of the cohomology of a K\"ahler manifold, the Hodge symmetry map is \emph{antilinear}: it identifies $\mathrm H^{p,q}$ with $\bar{\mathrm H}^{q,p}$.So far we have only discussed this map at the level of Hodge numbers. We asked whether any physical symmetry of the theory could recover this complex-antilinear operation.

In the context of physics, one has a meaningful inner product on the Hilbert space, which defines the notion of a vector of length one. However, as is familiar from undergraduate quantum mechanics, the phase of a normalized state can still in general be arbitrary.

In a unitary two-dimensional conformal field theory with Hilbert space \(\mathcal H\), consider a state \(\ket\phi\) with conformal weights \((h,\bar h)\). Identifying spacetime with the complex plane, the local operator \(\phi(z,\bar z)\) corresponding to \(\ket\phi\) via the operator-state correspondence can be Fourier-expanded as 
\begin{equation}
\phi(z,\bar z)=\sum_{n\in\mathbb Z+\eta}\sum_{\bar n\in\mathbb Z+\bar\eta}\phi_{-h-n,-\bar h-\bar n}z^n\bar z^{\bar n}
\end{equation}
where \(\eta,\bar\eta\in\{0,1/2\}\) according to specified boundary conditions. The modes' conformal weights take values in \((\mathbb Z+h+\eta,\mathbb Z+\bar h+\bar\eta)\). The Hermitian adjoints of these modes have conformal weights \((\mathbb Z-h-\eta,\mathbb Z-\bar h-\bar\eta)\). These two sets intersect (and coincide) precisely when \(h,\bar h\in\tfrac12\mathbb Z\). In that case, the following requirement may be nontrivially imposed:
\begin{equation}\label{conjugation}
\phi_{n,\bar n}^\dagger=\phi_{-n,-\bar n}.
\end{equation}
The set of states with (half-) integer conformal weights that satisfy \protect\cref{conjugation} form a real (but not complex) linear subspace of the entire Hilbert space. Call this subspace \(\mathcal H_0^{\mathbb R}\). Then its complexification
\[\mathcal H_0^{\mathbb R}\otimes_{\mathbb R}\mathbb C\equiv\mathcal H_0\subset\mathcal H\]
is the complex linear subspace of \(\mathcal H\) consisting of all states with (half-) integer conformal weights. By virtue of having a preferred real subspace, \(\mathcal H_0\) is equipped with a canonical \(\mathbb C\)-antilinear automorphism, which we denote\footnote{This construction is not at odds with the axiom of quantum mechanics that the (absolute) phase of a state is unobservable. The choice of operator-state map gives us the vacuum state (which maps to the identity operator), and phases of all states are measured relative to it. The phase of the vacuum (as determined by the operator-state correspondence) is arbitrary, as quantum mechanics requires.}
\[I\colon\mathcal H_0\to\mathcal H_0.\]
This map fixes the conformal weights, but otherwise all other charges are conjugated. In particular, for a theory with \(\mathcal N=(2,2)\) superconformal symmetry, \(I\) flips the various \(\operatorname U(1)\) charges:
\[I\colon(h,\bar h,q\Left,q\Right)\mapsto(h,\bar h,-q\Left,-q\Right).\]
Now, consider the parity-reversal operator \(\Omega\), which acts on the worldsheet coordinate \(z\) as
\[\Omega\colon z\mapsto\bar z.\]
This is a unitary operator that exchanges left- and right-moving symmetries. In particular,
\[\Omega\colon(h,\bar h,q\Left,q\Right)\mapsto(\bar h,h,q\Right,q\Left).\]
Composing the two, we obtain a \(\mathbb C\)-antilinear map
\[I\circ\Omega\colon(h,\bar h,q\Left,q\Right)\mapsto(\bar h,h,-q\Right,-q\Left).\]
It is easy to see that the above maps the (a,c)-ring to itself. In fact, in the A-model on a Calabi-Yau manifold, it is easy to see that the above acts on \((p,q)\)-forms as
\[(p,q)\mapsto(q,p),\]
since the degree \(p+q\) corresponds to the axial R-charge while \(p-q\) corresponds to the vector R-charge. 

Therefore, we have recovered the antilinear Hodge-symmetry map
\[\mathrm H^{p,q}\to\bar{\mathrm{H}}^{q,p}.\]

\section{Conclusion}
We wish to call the attention of physicists to the fact that many methods and results of rational homotopy theory are essentially quantum-mechanical in nature, and can be translated into physical language. Rational homotopy theory aims to study the geometry of manifolds as encoded in the de~Rham complex; as such, supersymmetric quantum mechanics as a probe of a space is its basic ingredient, and many of the main concepts that arise (Massey products, minimal models, and so on) are quite natural and pleasant to see from the perspective of physics. 

The techniques of rational homotopy theory allow one to extract the homotopy groups of a space, modulo torsion, from its de~Rham complex. We observed that identical techniques can be used to assign ``homotopy groups'' over~$\mathbb C$ to topologically twisted quantum field theories. If the theory in question satisfies an analogue of the $dd^c$-lemma, one can compute its homotopy groups just from the knowledge of the algebra of BPS states. In the simplest examples of topologically twistable sigma models, one recovers the ranks of higher homotopy groups of the target; beyond this setting, the meaning of these new invariants is unclear.

Arising out of this, one can use supersymmetric quantum mechanics as a tool to understand how close generic classes of quantum field theories are to being geometric---that is to say, describable as sigma models. 
There are certain properties that immediately mark a theory as being of non-geometric origin. For instance, the series of $\mathcal N=2$ superconformal minimal models are not geometric: their central charges, which would play the role of the dimension of the target space, are not integral. Similar integrality properties for $R$-charges (which translate in sigma models to the homological gradings) also fail. 

Work of Figueroa-O'Farrill, K\"ohl, and Spence~\cite{FOFKS} shows that these $\mathrm U(1)$ charges sometimes fit into a larger nonabelian symmetry that acts on the quantum mechanics. When this is the case, the required integrality properties follow immediately from integrality properties of weights for representations of semisimple symmetry algebras. In the context of sigma models, \cite{FOFKS} proved that these algebras (which are Lefschetz symmetries in that case) coincide with the algebras one expects to see from dimensional reduction of higher-dimensional sigma models; we pointed out that these symmetries arise in any quantum mechanics that is the dimensional reduction of a field theory, and compared them with the standard list of $R$-symmetries.

Regarding the integrality of the central charge, it is known that all representations of the two-dimensional $\mathcal N=(4,4)$ superconformal algebra have integer central charge. 
Furthermore, precisely in this case, the two $\mathrm U(1)_R$ charges fit into larger semisimple $R$-symmetries, and so also have integrality properties. This is at least circumstantial evidence that the rich structures that arise for hyper-K\"ahler sigma models are equivalent to the rich structure of $\mathcal N=(4,4)$ SCFTs, and makes it tempting to wonder whether both are equivalent to the property of coming from a six-dimensional theory. (This is obvious for $\mathcal N=(4,4)$ sigma models, which are the dimensional reductions of the six-dimensional sigma models defined on the same target space.)

For eight-dimensional hyper-K\"ahler manifolds, we observed a strange phenomenon: the Lefschetz symmetries that arise on their de~Rham complexes (which in other cases arise from dimensional reduction of sigma models) correspond to the values 10 and 26 of the worldsheet dimension. The effective target-space description of a string theory compactification on such manifolds would be two-dimensional. We have no way of understanding these numbers (and their occurrence may simply be a coincidence) but it would be interesting to search for connections between this phenomenon and the strange properties of supersymmetric quantum field theory in six dimensions. We defer this to future work.

Lastly, with an eye towards the same questions about geometric and non-geometric theories, we noted the physical relevance of results of Sullivan and others, related to the problem of reconstructing a manifold (up to diffeomorphism) from its de~Rham complex. In these results, it is important that one can begin with the de~Rham complex with rational or integer coefficients, rather than real or complex. We suggested that it may be possible to recover this data from $\mathcal N=(2,2)$ superconformal theories, subject to certain assumptions.
Understanding the appropriate analogues of theorems like Sullivan--Barge in the context of physics may lead to progress on questions such as mirror symmetry, which originate as certain ambiguities in the way that string theories encode geometric properties of the target space.
{\small
\vspace{-2ex}
\paragraph{Acknowledgments}
We thank B.~Knudsen, M.~Kolo\u glu, L.~McCloy, T.~McKinney, S.~Nawata, D.~Pei, and B.~Stoica for many fruitful conversations as these ideas were taking shape.
We are especially grateful to S.~Gukov and H.~Ooguri for comments, inspiration, and support.
Our work is funded by support from the Department of Energy (through grant {\rm{DE-SC}0011632}) and the Walter Burke Institute for Theoretical Physics at Caltech.
}
\printbibliography

\end{multicols}
\end{document}